\documentclass[10pt]{article}
\usepackage[margin=1.5in]{geometry}  
\usepackage{amsmath}
\usepackage{cases}
\usepackage{stmaryrd}
\usepackage{latexsym,amsfonts,amssymb,amsmath,amsthm}
\usepackage{mathrsfs}
\usepackage{graphicx}
\usepackage{verbatim}

\newtheorem{thm}{Theorem}[section]

\newtheorem{lem}[thm]{Lemma}

\newtheorem{defn}[thm]{Definition}
\newtheorem{rem}[thm]{Remark}

\numberwithin{equation}{section}

\def\C{\mathbb{C}}
\def\E{\mathbb{E}}
\def\N{\mathbb{N}}

\def\R{\mathbb{R}}

\def\Z{\mathbb{Z}}

\def\HH{\mathcal{H}}
\def\SS{\mathcal{S}}

\def\QQ{\mathcal{Q}}
\def\JJ{\mathcal{J}}
\def\KK{\mathcal{K}}

\def\lan{\langle}
\def\ran{\rangle}

\def\ra{\rightarrow}

\def\pa{\partial}

\def\al{\alpha}

\def\ep{\epsilon}

\def\tha{\theta}
\def\la{\lambda}
\def\La{\Lambda}
\def\si{\sigma}

\def\om{\omega}

\def\de{\delta}

\def\ga{\gamma}

\def\tha{\theta}
\def\Th{\Theta}

\begin{document}

\nocite{*}

\title{An Improved Combes-Thomas Estimate of Magnetic Schr\"{o}dinger Operators}

\author{Zhongwei Shen\footnote{Email: zzs0004@auburn.edu}\\Department of Mathematics and Statistics\\Auburn University\\Auburn, AL 36849\\USA}

\date{}

\maketitle

\begin{abstract}
In the present paper, we prove an improved Combes-Thomas estimate, i.e., the Combes-Thomas estimate in trace-class norms, for magnetic Schr\"{o}dinger operators under general assumptions. In particular, we allow unbounded potentials. We also show that for any function in the Schwartz space on the reals the operator kernel decays, in trace-class norms, faster than any polynomial.\\
Keywords: magnetic Schr\"{o}dinger operator, Combes-Thomas estimate, trace ideal estimate, operator kernel estimate.\\
2010 Mathematics Subject Classification: Primary 81Q10, 47F05; Secondary 35P05.
\end{abstract}


\section{Introduction}

The present paper is concerned with the so-called Combes-Thomas estimate of the following Schr\"{o}dinger operator with magnetic field
\begin{equation}\label{schrodinger-op}
H_{\La}(A,V)=\frac{1}{2}(-i\nabla-A(x))^{2}+V(x)\quad\text{on}\quad\La,
\end{equation}
where $i=\sqrt{-1}$ is the imaginary unit, $\nabla=(\pa_{x_{1}},\pa_{x_{2}},\dots,\pa_{x_{d}})$ is the gradient, $A$ is the vector potential giving rise to the magnetic field $\nabla\times A$, $V$ is the electric potential and $\La\subset\R^{d}$ is the configuration space with dimension $d$. This operator is used to characterize a spinless particle subject to a scaler potential and a magnetic filed in non-relativistic quantum physics \cite{GP90,GP91,Th02}.

As it is known, the Combes-Thomas estimate plays an important role in the theory of Schr\"{o}dinger operators, magnetic Schr\"{o}dinger operators, classical wave operators, etc. in random media. It was invented by Combes and Thomas \cite{CT73} to study the asymptotic behavior of eigenfunctions for multi-particle Schr\"{o}dinger operators. Later, Fr\"{o}hlich and Spencer \cite{FS83} used it to study the localization for the multidimensional discrete Anderson model. Meanwhile, the Combes-Thomas estimate, as well as Wegner estimate \cite{We81} and Lifshitz tail \cite{Lif65}, became
important ingredients in multiscale analysis. Specifically, the initial scale estimate in multiscale analysis for localization near the bottom of the spectrum is successful because of the Combes-Thomas estimate. See \cite{Ai94,BK05,CH94,FK96,FK97,GK01,GK03,GK04,Kir08,K08,KK01,KK04,Si82,St01} and references therein for further applications. Moreover, a stronger version of the Combes-Thomas estimate, i.e., the estimate in trace-class norms, is also very useful. In \cite{CHK07} and \cite{Ki10}, such estimates have been applied to study the regularity of the integrated density of states, a concept of great physical significance \cite{Pa73}. See \cite{BCH97,Kl95} for other applications.

Since the pioneering work of Combes and Thomas \cite{CT73}, the
Combes-Thomas estimate in operator norm has been well studied (see
\cite{Ai94,FK96,FK97,KK01,Si82,St01} and reference therein). We
point out the work of Germinet and Klein \cite{GK03}. They proved a
Combes-Thomas estimate, in operator norm, with explicit bound of
general Schr\"{o}dinger operators including Schr\"{o}dinger operator, magnetic Schr\"{o}dinger operator, acoustic operator, Maxwell operator and so on.  For the Combes-Thomas estimate in trace-class norms, existing results are scattered through the literature (see e.g. \cite{BNSS06},\cite{CH94},\cite{CHK07},\cite{Kl95}) and most of them  were proven (for special purposes), more or less, under additional assumptions. For instance, Klopp proved in \cite{Kl95} the estimate for Schr\"{o}dinger operators with bounded potentials. Barbaroux, Combes and Hislop's result, proven in \cite{BCH97} with an open spectrum gap assumption, works for a broad class of magnetic Schr\"{o}dinger operators, but was only proven for infinite-volume operators. Therefore, it is expected to obtain unified results for both finite-volume and infinite-volume magnetic Schr\"{o}dinger operators under general assumptions.

The main goal of the current paper is to obtain the Combes-Thomas
estimate of \eqref{schrodinger-op} and the associated operator
kernel estimate in trace-class norms under general assumptions,
which allow unbounded potentials. We first prove an
improved Combes-Thomas estimate, i.e., the Combes-Thomas estimate in
trace-class norms, for the magnetic Schr\"{o}dinger operator
\eqref{schrodinger-op} under general assumptions. Based on the
improved Combes-Thomas estimate, we then show that for any function
in the Schwartz space on the reals the operator kernel decays, in
trace-class norms, faster than any polynomial.

To be more specific, we assume that the magnetic vector potential
$A\in\HH_{loc}(\R^{d})$ is $\R^{d}$-valued, the electric potential $V\in\KK_{\pm}(\R^{d})$ is real-valued and the dimension $d\geq2$. The notations $\HH_{loc}(\R^{d})$ and $\KK_{\pm}(\R^{d})$ for spaces are explained in Section \ref{standing-notations}. Let $\La\subset\R^{d}$
be an open set. We assume that $\La$ is bounded with sufficiently smooth boundary if it is not the whole space. The self-adjoint realization of $H_{\La}(A,V)$ on $L^{2}(\La)$ is still denoted by $H_{\La}(A,V)$. If $\La\neq\R^{d}$, then $H_{\La}(A,V)$ is nothing but the localized operator with homogeneous Dirichlet boundary on $\pa\La$. These self-adjoint operators are constructed via sesquilinear forms. In Section \ref{section-semi-trace}, we will recall the constructions done in \cite{BHL00}.

Our first purpose is to study the Combes-Thomas estimate in trace
class norms, i.e., the trace ideal estimate of the operators
\begin{equation*}
\chi_{\beta}(H_{\La}(A,V)-z)^{-n}\chi_{\ga}, \quad\beta,\ga\in\R^{d},
\end{equation*}
where $\chi_{\beta}$ is the characteristic function of the unit cube
centered at $\beta\in\R^{d}$ and $z\in\rho(H_{\La}(A,V))$, the resolvent set of $H_{\La}(A,V)$. More precisely, we want to obtain the exponential decay of $\|\chi_{\beta}(H_{\La}(A,V)-z)^{-n}\chi_{\ga}\|_{\JJ_{p}}$ in terms of $|\beta-\ga|$ for suitable $n$ and $p$, where $\|\cdot\|_{\JJ_p}$ is the $p$-th von Neumann-Schatten norm reviewed in Section \ref{standing-notations}. Following the definition in \cite{GK03}, the family of operators $\{\chi_{\beta}(H_{\La}(A,V)-z)^{-n}\chi_{\ga}\}_{\beta,\ga\in\R^{d}}$ is also called the operator kernel of the bounded operator $(H_{\La}(A,V)-z)^{-n}$. In general, if $f$ is a bounded Borel function on $\si(H_{\La}(A,V))$, the spectrum of $H_{\La}(A,V)$, then the family $\{\chi_{\beta}f(H_{\La}(A,V))\chi_{\ga}\}_{\beta,\ga\in\R^{d}}$ is called the operator kernel of the bounded linear operator $f(H_{\La}(A,V))$. Our first main result regarding the Combes-Thomas estimate is roughly stated as follows (see Theorem \ref{theorem-decay-estimate} and Theorem \ref{corollary-decay-estimate} for details).

\begin{thm}\label{main-theorem}
Let $A\in\HH_{loc}(\R^{d})$, $V\in\KK_{\pm}(\R^{d})$ and $\La\subset\R^{d}$ open. Suppose $p>\frac{d}{2n}$ with $n\in\N$ and $n\geq1$. For any $z\in\rho(H_{\La}(A,V))$, the resolvent set of $H_{\La}(A,V)$, there exist constants $C=C(p,z,n)>0$ and $a_{0}=a_{0}(z)>0$ such that
\begin{equation*}
\|\chi_{\beta}(H_{\La}(A,V)-z)^{-n}\chi_{\ga}\|_{\JJ_{p}}\leq
Ce^{-a_{0}|\beta-\ga|},\quad\forall\,\,\beta,\ga\in\R^{d}.
\end{equation*}
\end{thm}

In this paper, we also study operator kernel estimate in trace-class norms. That is, we prove the polynomial decay of the operators
\begin{equation*}
\chi_{\beta}f(H_{\La}(A,V))\chi_{\ga},\quad\beta,\ga\in\R^{d}
\end{equation*}
in trace-class norms in terms of $|\beta-\ga|$, where $f$ belongs to the Schwartz space $\SS(\R)$ reviewed in Section \ref{standing-notations}. The main result related to operator kernel estimate is roughly stated as follows (see Theorem \ref{theorem-poly-decay-of-operator-kernel} for details).

\begin{thm}\label{main-theorem-poly}
Let $A\in\HH_{loc}(\R^{d})$, $V\in\KK_{\pm}(\R^{d})$ and $\La\subset\R^{d}$ open. Suppose $p>\frac{d}{2}$. Then, for any $f\in\SS(\R)$ and any $k\in\N$, there exists a constant $C=C(p,k,f)>0$ such that
\begin{equation}\label{main-estimate-polynomial-decay}
\|\chi_{\beta}f(H_{\La}(A,V))\chi_{\ga}\|_{\JJ_{p}}\leq C|\beta-\ga|^{-k},\quad\forall\,\,\beta,\ga\in\R^{d}.
\end{equation}
\end{thm}

Estimates like \eqref{main-estimate-polynomial-decay}, with $A$
being $\Z^{d}$-period, $V$ being bounded and $f$ being a smooth function with compact support, have been used, as a technical tool, to study the regularity of integrated density of states. For instance, Combes, Hislop and Klopp \cite[Eq.(2.30)]{CHK07} utilize the polynomial decay of any order to prove the convergence of some series, which leads to an expected estimate. It should be pointed out that Germinet and Klein proved in \cite{GK03} for slowly decreasing smooth functions (see Appendix \ref{app-Helffer-Sjostrand-formula} for the definition) the operator kernels for general Schr\"{o}dinger operators decay, in the operator norm, faster than any polynomial. Their result was then used as a crucial ingredient in their following paper \cite{GK04}. Later,  sub-exponential decay for functions in Gevrey classes and exponential decay for real analytic functions were obtained in \cite{BGK04} by Bouclet, Germinet and Klein.

The rest of the paper is organized as follows. In Section \ref{standing-notations}, we collect the notations used in this paper. In Section \ref{section-semi-trace}, we study trace ideal estimates of operators of the form $gf(H_{\Lambda}(A,V))$ for suitable $f$ and $g$. Such estimates, with $g$ being characteristic functions of unit cubes and $f$ being integer powers of the resolvent of $H_{\Lambda}(A,V)$, are used as technical tools in the proof of Theorem \ref{main-theorem}. Section
\ref{sec-combes-thomas-estimate} is devoted to the study of the Combes-Thomas estimate in trace-class norms. That is, we prove Theorem \ref{main-theorem}. In Section \ref{section-operator-kernel-estimate}, we study the operator kernel estimate in trace-class norms and prove Theorem
\ref{main-theorem-poly}.


\section{Standing Notations}\label{standing-notations}

In this section, we collect the notations which will be used in the sequel.

The configuration space $\La$ is an open set of $\R^{d}$. We assume that $\La$ is bounded with sufficiently smooth boundary unless it is the whole space. We also assume that the dimension $d\geq2$ since, by gauge transform, vector potentials in one spatial dimension are of no physical interest.

We denote by $\chi_{\beta}$ the characteristic function of the unit cube centered at $\beta\in\R^{d}$. If the configuration space in question is $\La(\neq\R^{d})$, then $\chi_{\beta}$ should be understood as $\chi_{\beta}\chi_{\La}$, where $\chi_{\La}$ is the characteristic function of $\La$. Generally speaking, if a function is defined on $\La$, then we consider it as a function defined on $\R^{d}$ by zero extension on $\R^{d}\backslash\La$.

The Banach space of $p$-th Lebesgue integrable functions on $\La$ is
\begin{equation*}
L^{p}(\La)=\big\{\phi\,\,\text{measurable on}\,\,\La\big|\|\phi\|_{p}<\infty\big\},
\end{equation*}
where
$\|\phi\|_{p}=\big(\int_{\La}|\phi(x)|^{p}dx\big)^{\frac{1}{p}}$ if $p\in[1,\infty)$ and $\|\phi\|_{\infty}=\text{ess sup}_{x\in\La}|\phi(x)|$. When $p=2$, $L^{2}(\La)$ is a Hilbert space with inner product
\begin{equation*}
\lan\phi,\psi\ran=\int_{\La}\bar{\phi}(x)\psi(x)dx.
\end{equation*}
Moreover, $\|\phi\|_{2}=\sqrt{\lan\phi,\phi\ran}$. As a convention, we simply write $\|\cdot\|_{2}$ as $\|\cdot\|$.

If $L:L^{p}(\La)\ra L^{q}(\La)$ is a bounded linear operator, the operator norm is defined by
\begin{equation*}
\|L\|_{p,q}:=\sup_{\|\phi\|_{p}=1}\|L\phi\|_{q}.
\end{equation*}
If $p=q=2$, we simply write $\|\cdot\|_{2,2}$ as $\|\cdot\|$.

Although we use the same notation $\|\cdot\|$ for both the norm of a function in $L^{2}(\La)$ and the norm of an operator on $L^{2}(\La)$, it should not give rise to any confusion. Similarly, we do not distinguish the notations for norms corresponding to different configuration spaces.

For any $p\in[1,\infty)$, the Banach space $\JJ_{p}$ (also an operator ideal) is defined by
\begin{equation*}
\JJ_{p}=\big\{C:L^{2}(\La)\ra L^{2}(\La)\,\,\text{linear and bounded}\big|\|C\|_{\JJ_{p}}<\infty\big\},
\end{equation*}
where $\|C\|_{\JJ_{p}}=\big(\text{Tr}|C|^{p}\big)^{\frac{1}{p}}<\infty$ is the $p$-th von Neumann-Schatten norm of $C$. See \cite{RS80,Si05} for more details. We here single out the space $\JJ_{2}$ (also called the space of Hilbert-Schmidt operators) for the following important property (see \cite[Theorem VI.23]{RS75}): a bounded linear operator $K$ on $L^{2}(\La)$ belongs to
$\mathcal{J}_{2}$ if and only if it is an integral operator with some integral kernel $k(x,y)$ being in
$L^{2}(\La\times\La)$. In this case, $\|K\|_{\JJ_{2}}=\big(\int_{\La\times\La}|k(x,y)|^{2}dxdy\big)^{\frac{1}{2}}$. We will use this property in Section \ref{section-semi-trace}.

Let $g(x)=-\ln|x|$ if $d=2$ and $g(x)=|x|^{2-d}$ if $d\geq3$. We say a function $V\in\KK(\R^{d})$, the Kato class, if
\begin{equation*}
\lim_{\ep\downarrow0}\sup_{x\in\R^{d}}\int_{|x-y|\leq\ep}g(x-y)|V(y)|dy=0.
\end{equation*}
A function $V$ is said to be in the local Kato class $\KK_{loc}(\R^{d})$ if $V\chi_{K}\in\KK(\R^{d})$ for all compact set $K\subset\R^{d}$, where $\chi_{K}$ is the characteristic function of $K$. We refer to \cite{Sz98} for equivalent definitions from the viewpoint of probability theory.

Let $V$ defined on $\R^{d}$ be real-valued. We say that $V$ is Kato decomposable, in symbols $V\in\KK_{\pm}(\R^{d})$, if the positive part $V_{+}$ is in $\KK_{loc}(\R^{d})$ and the negative part $V_{-}$ is in $\KK(\R^{d})$.

A $\C^{d}$-valued function $A$ is said to be in the class $\HH(\R^{d})$ if its squared norm $A\cdot A$ and its divergence $\nabla\cdot A$, considered as a distribution on $C_{0}^{\infty}(\R^{d})$, are both in the Kato class $\KK(\R^{d})$. It is said to be in the class $\HH_{loc}(\R^{d})$ if both $A\cdot A$ and $\nabla\cdot A$ are in the local Kato class $\KK_{loc}(\R^{d})$. We refer the reader to \cite{AS82,BHL00,CZ95,CFKS87} for further remarks about these spaces.

The Schwartz space $\SS(\R)$ consists of those $C^{\infty}(\R)$ functions which, together with all their derivatives, vanish at infinity faster than any power of $|x|$. More precisely, for any $N\in\Z$, $N\geq0$ and any $r\in\Z$, $r\geq0$, we define for $f\in C^{\infty}(\R)$
\begin{equation*}
\|f\|_{N,r}=\sup_{x\in\R}(1+|x|)^{N}|f^{(r)}(x)|,
\end{equation*}
then
\begin{equation*}
\SS(\R)=\{f\in C^{\infty}(\R)|\|f\|_{N,r}<\infty\,\,\text{for all}\,\,N,r\}.
\end{equation*}
See Folland \cite{Fo99} for more discussions about the Schwartz space.


\section{Semigroup and Trace Ideal Estimates}\label{section-semi-trace}

In this section, as a preparation for proving Theorem \ref{main-theorem} and Theorem \ref{main-theorem-poly}, we study estimates of operators of the form $gf(H_{\La}(A,V))$ in trace-class norms for suitable $f$ and $g$.

The self-adjoint realization of $H_{\La}(A,V)$ on $L^{2}(\La)$, still denoted by $H_{\La}(A,V)$, is defined via sesquilinear forms as follows (see \cite{BHL00}): the sesquilinear form
\begin{equation*}
\begin{split}
h_{\La}^{A,V_{+}}:&C_{0}^{\infty}(\La)\times C_{0}^{\infty}(\La)\ra\C,\\
&(\psi,\phi)\mapsto h_{\La}^{A,V_{+}}(\psi,\phi):=\big\lan\sqrt{V_{+}}\psi,\sqrt{V_{+}}\phi\big\ran+\frac{1}{2}\sum_{j=1}^{d}\big\lan(-i\pa_{j}-A_{j})\psi,(-i\pa_{j}-A_{j})\phi\big\ran
\end{split}
\end{equation*}
is densely defined in $L^{2}(\La)$, nonnegative and closable, where $\lan\cdot,\cdot\ran$ denotes the usual inner product on $L^{2}(\La)$. Its closure is still denoted by $h_{\La}^{A,V_{+}}$ with form domain $\QQ(h_{\La}^{A,V_{+}})$, which is the completion of $C_{0}^{\infty}(\La)$ with respect to the norm
\begin{equation*}
\|\phi\|_{h_{\La}^{A,V_{+}}}=\sqrt{\|\phi\|^{2}+h_{\La}^{A,V_{+}}(\phi,\phi)},
\end{equation*}
where $\|\cdot\|=\|\cdot\|_{2}$ is the norm on $L^{2}(\La)$ associated with $\lan\cdot,\cdot\ran$ as mentioned in Section \ref{standing-notations}. We denote by $H_{\La}(A,V_{+})$ the associated self-adjoint operator. Since $V_{-}\in\KK(\R^{d})$ is infinitesimally form-bounded with respective to $H_{\La}(A,0)(\leq H_{\La}(A,V_{+}))$, i.e., there exist $\Th_{1}\in(0,1)$ (can be taken to be arbitrarily small) and $\Th_{2}\geq0$ depending on $\Th_{1}$ so that
\begin{equation}\label{relative-form-bound}
\lan\phi,V_{-}\phi\ran\leq\Th_{1}h_{\La}^{A,0}(\phi,\phi)+\Th_{2}\|\phi\|^{2},\quad\phi\in\QQ(h_{\La}^{A,0}),
\end{equation}
KLMN theorem (see \cite[Theorem X.17]{RS75}) yields that, with $\QQ(h_{\La}^{A,V})=\QQ(h_{\La}^{A,V_{+}})$, the sesquilinear form
\begin{equation}\label{sesquilinear-form}
\begin{split}
h_{\La}^{A,V}:&\QQ(h_{\La}^{A,V})\times\QQ(h_{\La}^{A,V})\ra\C,\\
&(\psi,\phi)\mapsto
h_{\La}^{A,V}(\psi,\phi):=h_{\La}^{A,V_{+}}(\psi,\phi)-\big\lan\sqrt{V_{-}}\psi,\sqrt{V_{-}}\phi\big\ran
\end{split}
\end{equation}
is closed and bounded from below and has $C_{0}^{\infty}(\La)$ as a form core. The associated semi-bounded self-adjoint operator is denoted by $H_{\La}(A,V)$.

The main result of this section is stated as follows. Let
\begin{equation}\label{E-0}
E_{0}=\text{the infimum of the}\,\,L^{2}(\R^{d})\text{-spectrum of}\,\,H_{\R^{d}}(0,V).
\end{equation}

\begin{thm}\label{theorem-trace-ideal-estimate}
Let $A\in\HH_{loc}(\R^{d})$, $V\in\KK_{\pm}(\R^{d})$ and $\La\subset\R^{d}$ open. Suppose $p\geq2$. Let $f$ be a Borel function satisfying
\begin{equation}\label{assumption-f}
|f(\la)|\leq C(1+|\la|)^{-\al},\quad\la\in\si(H_{\La}(A,V)),
\end{equation}
for $\al>\frac{d}{2p}$. Then $gf(H_{\La}(A,V))$ is in $\JJ_{p}$ with
\begin{equation*}
\|gf(H_{\La}(A,V))\|_{\JJ_{p}}\leq
C_{\al,p,\la_{0}}\|g\|_{p}\|(H_{\La}(A,V)-\la_{0})^{\al}f(H_{\La}(A,V))\|
\end{equation*}
whenever $g\in L^{p}(\La)$, where $\la_{0}<E_{0}$ and $C_{\al,p,\la_{0}}>0$ depends only on $\al$, $p$ and $\la_{0}$.
\end{thm}

To prove the above theorem, we first present some lemmas. We begin with the celebrated Feynman-Kac-It\^{o} formula proven by Broderix, Hundertmark and Leschke (See \cite{Kir89,Si05A,Sz98} and references therein for earlier versions).
\begin{lem}[\cite{BHL00}]\label{lem-FK}
Let $A\in\HH_{loc}(\R^{d})$, $V\in\KK_{\pm}(\R^{d})$ and $\La\subset\R^{d}$ open. For any $\phi\in L^{2}(\La)$ and $t\geq0$, there holds
\begin{equation*}
\big(e^{-tH_{\La}(A,V)}\phi\big)(x)=\E_{x}\big\{e^{-S_{t}^{\om}(A,V)}\Xi_{\La,t}(\om)\phi(\om(t))\big\}\quad\text{for a.e.}\,\,x\in\La,
\end{equation*}
where
\begin{equation*}
S_{t}^{\om}(A,V)=i\int_{0}^{t}A(\om(s))d\om(s)+\frac{i}{2}\int_{0}^{t}(\nabla\cdot
A)(\om(s))ds+\int_{0}^{t}V(\om(s))ds,
\end{equation*}
$\E_{x}\{\cdot\}$ denotes the expectation for the Brownian motion starting at $x$ and $\Xi_{\La,t}$ is the characteristic function of the set $\{\om|\om(s)\in\La\,\,\text{for all}\,\,s\in[0,t]\}$.
\end{lem}
As consequences of Lemma \ref{lem-FK}, we get the so called diamagnetic inequality
\begin{equation*}
\big|e^{-tH_{\La}(A,V)}\phi\big|\leq e^{-tH_{\La}(0,V)}|\phi|,\quad t\geq0,
\end{equation*}
the monotonicity of semigroup for vanishing magnetic field in the sense that for $\La\subset\La'$
\begin{equation*}\label{mono-semigroup}
e^{-tH_{\La}(0,V)}\chi_{\La}\phi\leq e^{-tH_{\La'}(0,V)}\phi,\quad\phi\geq0,\,\,t\geq0
\end{equation*}
and then the $L^{p}$-smoothing of semigroups: for $1\leq p\leq q\leq\infty$, there exist constant $C>0$ and $E$ such that
\begin{equation}\label{smoothing-of-semigroup}
\big\|e^{-tH_{\La}(A,V)}\big\|_{p,q}\leq\big\|e^{-tH_{\La}(0,V)}\big\|_{p,q}\leq\big\|e^{-tH_{\R^{d}}(0,V)}\big\|_{p,q}\leq
Ct^{-\ga}e^{Et},
\end{equation}
where $\ga=\frac{d}{2}(\frac{1}{p}-\frac{1}{q})$. We remark that $E$ can be chosen such that $-E<E_{0}$ (see e.g. \cite{BHL00,Si82}).

We extend \cite[Theorem B.2.1]{Si82} to the magnetic case.

\begin{lem}\label{lemma-resolvent-estimate}
Let $A\in\HH_{loc}(\R^{d})$, $V\in\KK_{\pm}(\R^{d})$ and $\La\subset\R^{d}$ open. Let $\al>0$ and $1\leq p\leq q\leq\infty$ satisfy
\begin{equation}\label{assumption-1}
\frac{1}{p}-\frac{1}{q}<\frac{2\al}{d}.
\end{equation}
Then $(H_{\La}(A,V)-z)^{-\al}$ is bounded from $L^{p}(\La)$ to $L^{q}(\La)$ whenever the real part $\Re
z<E_{0}$.
\end{lem}
\begin{proof}
It follows from the formula
\begin{equation*}
(H_{\La}(A,V)-z)^{-\al}=c_{\al}\int_{0}^{\infty}e^{-tH_{\La}(A,V)}e^{tz}t^{\al-1}dt
\end{equation*}
and \eqref{smoothing-of-semigroup}, where the assumption \eqref{assumption-1} is applied to insure the convergence of the above integral.
\end{proof}

As a consequence of Lemma \ref{lemma-resolvent-estimate}, we have

\begin{lem}\label{coro-1}
Let $A\in\HH_{loc}(\R^{d})$, $V\in\KK_{\pm}(\R^{d})$ and $\La\subset\R^{d}$ open. Let $\al>0$ and $1\leq p\leq2\leq q\leq\infty$ satisfy \eqref{assumption-1}. For any Borel function $f$ satisfying \eqref{assumption-f}, the operator $f(H_{\La}(A,V))$ is bounded from $L^{p}(\La)$ to $L^{q}(\La)$ with
\begin{equation*}
\|f(H_{\La}(A,V))\|_{p,q}\leq C_{p,q,\al,\la_{0}}\|(H_{\La}(A,V)-\la_{0})^{\al}f(H_{\La}(A,V))\|,
\end{equation*}
where $\la_{0}<E_{0}$ and $C_{p,q,\al,\la_{0}}>0$ depends only on $p$, $q$, $\al$ and $\la_{0}$.
\end{lem}
\begin{proof}
It follows from the arguments in \cite[Theorem B.2.3]{Si82}.
\end{proof}

We next discuss the trace ideal estimate of operators of the form $gf(H_{\La}(A,V))$ for suitable $f$ and $g$. We start with recalling a result of Dunford and Pettis (See \cite{CFKS87,Si82,Tr67} for abstract versions).

\begin{lem}\label{theorem-Dunford-Pettis}
Let $(M,\mu)$ be a separable measurable space. If $L$ is a bounded linear operator from $L^{p}(M)$ to $L^{\infty}(M)$ with $1\leq p<\infty$, then there is a measurable function $k(\cdot,\cdot)$ on $M\times M$ such that $L$ is an integral operator with integral kernel $k(\cdot,\cdot)$ and
\begin{equation*}
\sup_{x\in
M}\bigg(\int_{M}|k(x,y)|^{p'}d\mu(y)\bigg)^{\frac{1}{p'}}=\|L\|_{p,\infty}<\infty,
\end{equation*}
where $p'=\frac{p}{p-1}$ is the conjugate exponent of $p$.
\end{lem}

We are now ready to prove Theorem \ref{theorem-trace-ideal-estimate}.

\begin{proof}[Proof of Theorem \ref{theorem-trace-ideal-estimate}]
By complex interpolation (see \cite[Theorem 2.9]{Si05}), it suffices to prove the result in the case $p=2$, which we show now. For $p=2$ and $q=\infty$, we have $\frac{d}{2}\big(\frac{1}{p}-\frac{1}{q}\big)=\frac{d}{4}<\al$ by assumption, i.e., \eqref{assumption-1} is satisfied, and thus, Lemma \ref{coro-1} implies that $f(H_{\La}(A,V))$ is bounded from $L^{2}(\La)$ to $L^{\infty}(\La)$. By Lemma \ref{theorem-Dunford-Pettis}, $f(H_{\La}(A,V))$ is an integral operator with kernel $k_{\La}^{A,V}(x,y)$ satisfying
\begin{equation*}
\sup_{x\in\La}\int_{\La}\big|k_{\La}^{A,V}(x,y)\big|^{2}dy=\|f(H_{\La}(A,V))\|_{2,\infty}^{2}<\infty.
\end{equation*}
Thus, $gf(H_{\La}(A,V))$ is an integral operator on $L^{2}(\La)$ with kernel $g(x)k_{\La}^{A,V}(x,y)$. Moreover,
\begin{equation*}
\iint_{\La\times\La}\big|g(x)k_{\La}^{A,V}(x,y)\big|^{2}dxdy\leq\|g\|_{2}^{2}\sup_{x\in\La}\int_{\La}\big|k_{\La}^{A,V}(x,y)\big|^{2}dy=\|g\|_{2}^{2}\|f(H_{\La}(A,V))\|_{2,\infty}^{2},
\end{equation*}
which implies that $gf(H_{\La}(A,V))$ is a Hilbert-Schmidt operator as mentioned in Section \ref{standing-notations}, i.e., in $\JJ_{2}$, with $\JJ_{2}$-norm bounded by $\|g\|_{2}\|f(H_{\La}(A,V))\|_{2,\infty}$. The expected bound is given by Lemma \ref{coro-1}. This completes the proof.
\end{proof}

We remark that results obtained in this section are well-known for Schr\"{o}dinger operators without magnetic fields. See \cite{AS82,Si82} and references therein. It should be pointed out that the result of Theorem \ref{theorem-trace-ideal-estimate} in the case $H_{\R^{d}}(0,V)$ was proven in \cite[Theorem
B.9.3]{Si82} for any $p\geq1$. To prove the result for $p\in[1,2)$, it was first shown that
$gf(H_{\R^{d}}(0,V))\in\JJ_{1}$ for $g\in\ell^{1}(L^{2}(\R^{d}))$, the Birman-Solomjak space, then proceeded to complex interpolation. The proof relies on the translation invariance of the free Laplacian (see \cite[Theorem B.9.2]{Si82} and \cite[Theorem 4.5]{Si05} for instance), which, however, is not true for magnetic Schr\"{o}dinger operators. This prevents us from obtaining the result for $p\in[1,2)$.


\section{The Combes-Thomas Estimate in Trace Ideals}\label{sec-combes-thomas-estimate}

In this section, we study the improved Combes-Thomas estimate, i.e., the trace ideal estimate of the operators
\begin{equation*}
\chi_{\beta}(H_{\La}(A,V)-z)^{-n}\chi_{\ga}\quad\text{for}\,\,\beta,\ga\in\R^{d},
\end{equation*}
where $\chi_{\beta}$ is the characteristic function of the unit cube centered at $\beta$. More precisely, we want to obtain the exponential decay of $\|\chi_{\beta}(H_{\La}(A,V)-z)^{-1}\chi_{\ga}\|_{\JJ_{p}}$ in terms of $|\beta-\ga|$. The main result is stated in Theorem \ref{main-theorem}. Since we also consider localized operators, $\chi_{\beta}$ should be understood as $\chi_{\beta}\chi_{\La}$ if the operators is restricted to $\La$ as it is mentioned in Section \ref{standing-notations}, where $\chi_{\La}$ is the characteristic function of the domain $\La$. The basic tools we use here are sectorial form and $m$-sectorial operator reviewed in
Appendix \ref{app-sectorial-form}. We also employ the classical argument of Combes and Thomas developed in \cite{CT73}.

First of all, we establish some results by applying the theory of sectorial form and $m$-sectorial operator. For this purpose, we first define auxiliary sesquilinear forms with associated operators
formally given by
\begin{equation}\label{auxiliary-operator}
H_{\La}^{a}(A,V)=e^{a\cdot x}H_{\La}(A,V)e^{-a\cdot x},\quad a\in\R^{d},
\end{equation}
where $e^{a\cdot x}$ and $e^{-a\cdot x}$ are multiplicative operators. Note that the operator $H_{\La}^{a}(A,V)$ is not self-adjoint unless $a=0$. First, we denote by $D_{A,\La}$ the closure of $\frac{\sqrt{2}}{2}(-i\nabla-A)$ on $C_{0}^{\infty}(\La)$, so $H_{\La}(A,0)=D_{A,\La}^{*}D_{A,\La}$. This can be seen by sesquilinear forms. Moreover, the domain of $D_{A,\La}$, denoted by $\mathscr{D}(D_{A,\La})$, is the form domain, denoted by $\QQ(h_{\La}^{A,0})$, of the sesquiliner form
associated with the lower bounded self-adjoint operator $H_{\La}(A,0)$. For $a\in\La$, we define
\begin{equation*}
D_{A,\La}(a)=e^{a\cdot x}D_{A,\La}e^{-a\cdot x}\quad\text{and}\quad D_{A,\La}^{*}(a)=e^{a\cdot
x}D_{A,\La}^{*}e^{-a\cdot x}.
\end{equation*}
It's easy to see that
\begin{equation}\label{two-equalities}
\begin{split}
D_{A,\La}(a)&=D_{A,\La}+i\frac{\sqrt{2}}{2}a,\quad\text{on}\quad\mathscr{D}(D_{A,\La}),\\
D_{A,\La}^{*}(a)&=D_{A,\La}^{*}+i\frac{\sqrt{2}}{2}a,\quad\text{on}\quad\mathscr{D}(D_{A,\La}^{*})
\end{split}
\end{equation}
and they are closed, densely defined operators. Note $(D_{A,\La}(a))^{*}\neq D_{A,\La}^{*}(a)$ for $a\neq0$. Next, we define the sesquilinear form $h_{\La}^{A,0}(a)$ on $\mathscr{D}(D_{A,\La})=\QQ(h_{\La}^{A,0})$ by
\begin{equation}\label{sesquilinear-form-aux-0}
h_{\La}^{A,0}(a)(\psi,\phi)=\big\lan(D_{A,\La}^{*}(a))^{*}\psi,D_{A,\La}(a)\phi\big\ran.
\end{equation}
Obviously, $h_{\La}^{A,0}(0)\equiv h_{\La}^{A,0}$. Finally, we define the sesquilinear form $h_{\La}^{A,V}(a)$ on $\QQ(h_{\La}^{A,V_{+}})$ by
\begin{equation}\label{sesquilinear-form-aux}
h_{\La}^{A,V}(a)(\psi,\phi)=h_{\La}^{A,0}(a)(\psi,\phi)+\big\lan\sqrt{V_{+}}\psi,\sqrt{V_{+}}\phi\big\ran-\big\lan
\sqrt{V_{-}}\psi,\sqrt{V_{-}}\phi\big\ran.
\end{equation}

For $a_{0}>0$, let
\begin{equation}\label{two-constants}
\begin{split}
\Xi_{1}(s)=\frac{2s}{1-\Th_{1}},\quad\Xi_{2}(s,a_{0})=\frac{2s\Th_{2}}{1-\Th_{1}}+\bigg(\frac{1}{2s}+\frac{s}{4}\bigg)a_{0}^{2},
\end{split}
\end{equation}
where $\Th_{1}$, $\Th_{2}$ are given in \eqref{relative-form-bound}. We will write $\Xi_{1}(s)$ and $\Xi_{2}(s,a_{0})$ as $\Xi_{1}$ and $\Xi_{2}$, respectively, in the sequel.

We next prove several lemmas related to $H_{\La}^{a}(A,V)$. Our first lemma is about the relation between $h_{\La}^{A,V}(a)$ and $H_{\La}^{a}(A,V)$.

\begin{lem}\label{lemma-sectorial}
Let $A\in\HH_{loc}(\R^{d})$, $V\in\KK_{\pm}(\R^{d})$ and $\La\subset\R^{d}$ open. The sesquilinear form
$h_{\La}^{A,V}(a)$ defined in \eqref{sesquilinear-form-aux} is a closed sectorial form associated with the unique $m$-sectorial operator $H_{\La}^{a}(A,V)$ given by \eqref{auxiliary-operator}.
\end{lem}
\begin{proof}
By \eqref{sesquilinear-form}, \eqref{two-equalities}, \eqref{sesquilinear-form-aux-0} and $\eqref{sesquilinear-form-aux}$, we have for any $\phi\in\QQ(h_{\La}^{A,V})$,
\begin{equation*}
\big|h_{\La}^{A,V}(a)(\phi,\phi)-h_{\La}^{A,V}(\phi,\phi)\big|=\big|h_{\La}^{A,0}(a)(\phi,\phi)-h_{\La}^{A,0}(\phi,\phi)\big|\leq\sqrt{2}|\Re\lan\phi,a\cdot
D_{A,\La}\phi\ran|+\frac{1}{2}|a|^{2}\|\phi\|^{2}
\end{equation*}
so that
\begin{equation*}
\big|h_{\La}^{A,V}(a)(\phi,\phi)-h_{\La}^{A,V}(\phi,\phi)\big|^{2}\leq4|a|^{2}\|\phi\|^{2}\|D_{A,\La}\phi\|^{2}+\frac{1}{2}|a|^{4}\|\phi\|^{4},
\end{equation*}
which implies that for any $s>0$, 
\begin{equation}\label{estimate-4}
\begin{split}
&\big|h_{\La}^{A,V}(a)(\phi,\phi)-h_{\La}^{A,V}(\phi,\phi)\big|\\
&\quad\quad\leq|a|\|\phi\|\bigg(4\|D_{A,\La}\phi\|^{2}+\frac{1}{2}|a|^{2}\|\phi\|^{2}\bigg)^{\frac{1}{2}}\\
&\quad\quad\leq\frac{1}{2s}|a|^{2}\|\phi\|^{2}+\frac{s}{2}\bigg(4\|D_{A,\La}\phi\|^{2}+\frac{1}{2}|a|^{2}\|\phi\|^{2}\bigg)\\
&\quad\quad=2sh_{\La}^{A,0}(\phi,\phi)+\bigg(\frac{1}{2s}+\frac{s}{4}\bigg)|a|^{2}\|\phi\|^{2},
\end{split}
\end{equation}
since $h_{\La}^{A,0}(\phi,\phi)=\|D_{A,\La}\phi\|^{2}$. Thanks to \eqref{relative-form-bound} and \eqref{sesquilinear-form},
\begin{equation*}
h_{\La}^{A,V}\geq(1-\Th_{1})h_{\La}^{A,0}-\Th_{2}\quad\text{on}\quad\QQ(h_{\La}^{A,V})\big(\subset\QQ(h_{\La}^{A,0})\big).
\end{equation*}
This, together with \eqref{estimate-4}, implies that
\begin{equation}\label{relative-bounded-estimate}
\big|h_{\La}^{A,V}(a)(\phi,\phi)-h_{\La}^{A,V}(\phi,\phi)\big|\leq\Xi_{1}h_{\La}^{A,V}(\phi,\phi)+\Xi_{2}\|\phi\|^{2},
\quad\phi\in\QQ(h_{\La}^{A,V}),
\end{equation}
where $\Xi_{1}$ and $\Xi_{2}$ are given in \eqref{two-constants} with $a_{0}$ replaced by $|a|$.

To apply Theorem \ref{relatively-bounded-pert}, we choose $s\in\big(0,\frac{1-\Th_{1}}{2}\big)$ so that
$\Xi_{1}=\frac{2s}{1-\Th_{1}}<1$. Since $h_{\La}^{A,V}$ is symmetric, closed and bounded from below, Theorem \ref{relatively-bounded-pert} says that $h_{\La}^{A,V}(a)$ is a closed sectorial form defined on $\QQ(h_{\La}^{A,V})$. Theorem \ref{representation-theorem} then guarantees that there exists a unique $m$-sectorial operator, denoted by $H_{\La}^{a}(A,V)$, associated to $h_{\La}^{A,V}(a)$.
\end{proof}

The next lemma gives an operator equality connecting $H_{\La}^{a}(A,V)$ and $H_{\La}(A,V)$.

\begin{lem}\label{lemma-representation}
Let $A\in\HH_{loc}(\R^{d})$, $V\in\KK_{\pm}(\R^{d})$ and $\La\subset\R^{d}$ open. Suppose
$s\in\big(0,\frac{1-\Th_{1}}{2}\big)$ so that $\Xi_{1}<1$. Set
\begin{equation}\label{associated-operator-2}
\tilde{H}_{\La}(A,V)=H_{\La}(A,V)+\Xi_{1}^{-1}\Xi_{2},
\end{equation}
where $\Xi_{1}$ and $\Xi_{2}$ are given in \eqref{two-constants} with $a_{0}$ replaced by $|a|$. Then $\tilde{H}_{\La}(A,V)$ is nonnegative and there exists a bounded linear operator $B$ from $L^{2}(\La)$ to itself with $\|B\|\leq2\Xi_{1}$ such that
\begin{equation}\label{expected-equality}
H_{\La}^{a}(A,V)=H_{\La}(A,V)+\sqrt{\tilde{H}_{\La}(A,V)}B\sqrt{\tilde{H}_{\La}(A,V)},
\end{equation}
where $H_{\Lambda}^{a}(A,V)$ is the $m$-sectorial operator in Lemma
\ref{lemma-sectorial}.
\end{lem}
\begin{proof}
Set
\begin{equation}\label{new-sesquilinear-forms}
\begin{split}
\bar{h}_{\La}^{A,V}(a)&=h_{\La}^{A,V}(a)-h_{\La}^{A,V}\quad\text{on}\quad\QQ(h_{\La}^{A,V}),\\
\tilde{h}_{\La}^{A,V}&=h_{\La}^{A,V}+\Xi_{1}^{-1}\Xi_{2}\quad\text{on}\quad\QQ(h_{\La}^{A,V}).
\end{split}
\end{equation}
Then \eqref{relative-bounded-estimate} can be rewritten as
\begin{equation*}
\big|\bar{h}_{\La}^{A,V}(a)(\phi,\phi)\big|\leq\Xi_{1}\tilde{h}_{\La}^{A,V}(\phi,\phi),\quad\phi\in\QQ(h_{\La}^{A,V}),
\end{equation*}
which implies that $\tilde{h}_{\La}^{A,V}$ is a densely defined, symmetric, nonnegative closed sesquilinear form with the associated nonnegative self-adjoint operator $\tilde{H}_{\La}(A,V)$ defined
in \eqref{associated-operator-2}.

Theorem \ref{representation-theorem-1} then insures that there exists a bounded linear operator $B$ from $L^{2}(\La)$ to itself with $\|B\|\leq2\Xi_{1}$ so that
\begin{equation}\label{representation-of-new-form}
\bar{h}_{\La}^{A,V}(a)(\psi,\phi)=\bigg\lan
\sqrt{\tilde{H}_{\La}(A,V)}\psi,B\sqrt{\tilde{H}_{\La}(A,V)}\phi\bigg\ran
\end{equation}
for $\psi,\phi\in\QQ(h_{\La}^{A,V})=\mathscr{D}\Big(\sqrt{\tilde{H}_{\La}(A,V)}\Big)$. Let
\begin{equation}\label{new-form-1}
\tilde{h}_{\La}^{A,V}(a)=h_{\La}^{A,V}(a)+\Xi_{1}^{-1}\Xi_{2}\quad\text{on}\quad\QQ(h_{\La}^{A,V}).
\end{equation}
Since $h_{\La}^{A,V}(a)$ is a densely defined closed sectorial form, so does $\tilde{h}_{\La}^{A,V}(a)$ and the associated $m$-sectorial operator is given by
\begin{equation}\label{associated-operator-1}
\tilde{H}_{\La}^{a}(A,V)=H_{\La}^{a}(A,V)+\Xi_{1}^{-1}\Xi_{2}.
\end{equation}
Considering \eqref{new-sesquilinear-forms} and \eqref{representation-of-new-form}, we also have
\begin{equation}\label{another-form}
\begin{split}
&\tilde{h}_{\La}^{A,V}(a)(\psi,\phi)\\
&\quad\quad=\tilde{h}_{\La}^{A,V}(\psi,\phi)+\bar{h}_{\La}^{A,V}(a)(\psi,\phi)\\
&\quad\quad=\bigg\lan\sqrt{\tilde{H}_{\La}(A,V)}\psi,\sqrt{\tilde{H}_{\La}(A,V)}\phi\bigg\ran+\bigg\lan\sqrt{\tilde{H}_{\La}(A,V)}\psi,B\sqrt{\tilde{H}_{\La}(A,V)}\phi\bigg\ran\\
&\quad\quad=\bigg\lan\sqrt{\tilde{H}_{\La}(A,V)}\psi,(1+B)\sqrt{\tilde{H}_{\La}(A,V)}\phi\bigg\ran,\quad\psi,\phi\in\QQ(h_{\La}^{A,V}).
\end{split}
\end{equation}

We claim that
\begin{equation}\label{same-operator}
\tilde{H}_{\La}^{a}(A,V)=\sqrt{\tilde{H}_{\La}(A,V)}(1+B)\sqrt{\tilde{H}_{\La}(A,V)}.
\end{equation}
Let
$\phi\in\mathscr{D}(\tilde{H}_{\La}^{a}(A,V))\subset\QQ(\tilde{h}_{\La}^{A,V}(a))=\QQ(h_{\La}^{A,V})$.
We have
\begin{equation*}
\tilde{h}_{\La}^{A,V}(a)(\psi,\phi)=\lan\psi,\tilde{H}_{\La}^{a}(A,V)\phi\ran\,\,\text{for all}\,\,\psi\in\QQ(\tilde{h}_{\La}^{A,V}(a))=\QQ(h_{\La}^{A,V}).
\end{equation*}
Comparing this with \eqref{another-form} and recalling the definition of the adjoint of an operator, we see that $\sqrt{\tilde{H}_{\La}(A,V)}(1+B)\sqrt{\tilde{H}_{\La}(A,V)}\phi$ exists and is equal to $\tilde{H}_{\La}^{a}(A,V)\phi$, which
implies that
\begin{equation*}
\tilde{H}_{\La}^{a}(A,V)\subset\sqrt{\tilde{H}_{\La}(A,V)}(1+B)\sqrt{\tilde{H}_{\La}(A,V)},
\end{equation*}
i.e.,
$\sqrt{\tilde{H}_{\La}(A,V)}(1+B)\sqrt{\tilde{H}_{\La}(A,V)}$ extends $\tilde{H}_{\La}^{a}(A,V)$. To show \eqref{same-operator}, it now suffices to show that $\sqrt{\tilde{H}_{\La}(A,V)}(1+B)\sqrt{\tilde{H}_{\La}(A,V)}$ is accretive since $\tilde{H}_{\La}^{a}(A,V)$ is $m$-sectorial, so has no proper accretive extension. For any
$\psi\in\mathscr{D}\bigg(\sqrt{\tilde{H}_{\La}(A,V)}(1+B)\sqrt{\tilde{H}_{\La}(A,V)}\bigg)\subset\QQ(h_{\La}^{A,V})$,
\eqref{another-form} and \eqref{new-form-1} give
\begin{equation*}
\begin{split}
\bigg\lan
&\psi,\sqrt{\tilde{H}_{\La}(A,V)}(1+B)\sqrt{\tilde{H}_{\La}(A,V)}\psi\bigg\ran\\
&\quad=\tilde{h}_{\La}^{A,V}(a)(\psi,\psi)\\
&\quad=h_{\La}^{A,V}(a)(\psi,\psi)+\Xi_{1}^{-1}\Xi_{2}\|\psi\|^{2}\\
&\quad=h_{\La}^{A,V}(a)(\psi,\psi)-h_{\La}^{A,V}(\psi,\psi)+h_{\La}^{A,V}(\psi,\psi)+\Xi_{1}^{-1}\Xi_{2}\|\psi\|^{2}.
\end{split}
\end{equation*}
It then follows from
\begin{equation*}
\big|\Re\big(h_{\La}^{A,V}(a)(\psi,\psi)-h_{\La}^{A,V}(\psi,\psi)\big)\big|\leq\big|h_{\La}^{A,V}(a)(\psi,\psi)-h_{\La}^{A,V}(\psi,\psi)\big|
\end{equation*}
and \eqref{relative-bounded-estimate} that
\begin{equation*}
\begin{split}
\Re\bigg\lan
&\psi,\sqrt{\tilde{H}_{\La}(A,V)}(1+B)\sqrt{\tilde{H}_{\La}(A,V)}\psi\bigg\ran\\
&\quad=\Re(h_{\La}^{A,V}(a)(\psi,\psi)-h_{\La}^{A,V}(\psi,\psi))+h_{\La}^{A,V}(\psi,\psi)+\Xi_{1}^{-1}\Xi_{2}\|\psi\|^{2}\\
&\quad\geq-(\Xi_{1}h_{\La}^{A,V}(\psi,\psi)+\Xi_{2}\|\psi\|^{2})+h_{\La}^{A,V}(\psi,\psi)+\Xi_{1}^{-1}\Xi_{2}\|\psi\|^{2}\\
&\quad=(1-\Xi_{1})(h_{\La}^{A,V}(\psi,\psi)+\Xi_{1}^{-1}\Xi_{2}\|\psi\|^{2})\\
&\quad\geq0,
\end{split}
\end{equation*}
since $\Xi_{1}$ is taken to be less than $1$ and $h_{\La}^{A,V}+\Xi_{1}^{-1}\Xi_{2}$ is nonnegative by
\eqref{relative-bounded-estimate}. This shows that $\sqrt{\tilde{H}_{\La}(A,V)}(1+B)\sqrt{\tilde{H}_{\La}(A,V)}$ is accretive and, thus, \eqref{same-operator} holds. Obviously, \eqref{expected-equality} is equivalent to \eqref{same-operator} due to \eqref{associated-operator-2} and \eqref{associated-operator-1}. This completes the proof.
\end{proof}

The last lemma bridges the resolvent set of $H_{\La}(A,V)$ and that of $H_{\La}^{a}(A,V)$. Before stating the result, we make following assumptions.

Pick and fix $\la_{0}<\min\{-\Th_{2},E_{0}\}$, where $E_{0}$ is defined in \eqref{E-0}. This number is picked to be of technical use. The main advantage is that $H_{\La}(A,V)-\la_{0}$ is strictly positive so that $(H_{\La}(A,V)-\la_{0})^{\frac{1}{2}}$ is well-defined and boundedly invertible, as opposed to the ill-posedness of the fractional power of $H_{\La}(A,V)-z$, which may cause some troubles.

Let
\begin{equation*}
c_{z,\la_{0}}=\bigg\|\frac{\la-\la_{0}}{\la-z}\bigg\|_{L^{\infty}(\si(H_{\La}(A,V)))}.
\end{equation*}
Suppose that $s>0$ and $a_{0}>0$ satisfy
\begin{equation}\label{condition-on-a-1}
s<\frac{1-\Th_{1}}{4c_{z,\la_{0}}}\quad\text{and}\quad a_{0}^{2}\leq\frac{2s(\la_{0}+\Th_{2})}{\Th_{1}-1}\bigg(\frac{1}{2s}+\frac{s}{4}\bigg)^{-1}
\end{equation}
or
\begin{equation}\label{condition-on-a-2}
\frac{2s(\la_{0}+\Th_{2})}{\Th_{1}-1}\bigg(\frac{1}{2s}+\frac{s}{4}\bigg)^{-1}\leq a_{0}^{2}<\bigg(\frac{(\de-1)\la_{0}}{2c_{z,\la_{0}}}+\frac{2s(\de\la_{0}+\Th_{2})}{\Th_{1}-1}\bigg)\bigg(\frac{1}{2s}+\frac{s}{4}\bigg)^{-1},
\end{equation}
where $\de=\de(\la_{0})\in(0,1)$ is such that $\de\la_{0}\in\big(\la_{0},\min\{-\Th_{2}, E_{0}\}\big)$. We will show the derivation of the above two classes of conditions in Lemma \ref{condition-invertibility} below. We point out that
\begin{equation*}
\frac{2s(\la_{0}+\Th_{2})}{\Th_{1}-1}<\frac{(\de-1)\la_{0}}{2c_{z,\la_{0}}}+\frac{2s(\de\la_{0}+\Th_{2})}{\Th_{1}-1}
\end{equation*}
is nothing but $s<\frac{1-\Th_{1}}{4c_{z,\la_{0}}}$.

\begin{rem}\label{remark-main-theorem}
Note that assumptions \eqref{condition-on-a-1} and \eqref{condition-on-a-2} can be considered together to form a more general one, but we consider them separately anyway for the following two reasons.
\begin{itemize}
\item[(i)] The first reason is about the conditions giving rise to \eqref{condition-on-a-1} and the first inequality in \eqref{condition-on-a-2}. In the proof of Lemma \ref{lemma-invertibility-of-auxiliary-operator} below, we need conditions on $s$ and $a_{0}$ to insure
\begin{equation*}
2\Xi_{1}c_{z,\la_{0}}\bigg\|\frac{\la+\Xi_{1}^{-1}\Xi_{2}}{\la-\la_{0}}\bigg\|_{L^{\infty}(\si(H_{\La}(A,V)))}<1,
\end{equation*}
i.e.,\eqref{expected-condition}, where the quantity
$\big\|\frac{\la+\Xi_{1}^{-1}\Xi_{2}}{\la-\la_{0}}\big\|_{L^{\infty}(\si(H_{\La}(A,V)))}$ appears. It's easy to see that
\begin{equation*}
\bigg\|\frac{\la+\Xi_{1}^{-1}\Xi_{2}}{\la-\la_{0}}\bigg\|_{L^{\infty}(\si(H_{\La}(A,V)))}=
\left\{
\begin{aligned}
1,\quad\text{if}\quad\Xi_{1}^{-1}\Xi_{2}\leq-\la_{0},\\
\frac{\inf\si(H_{\La}(A,V))+\Xi_{1}^{-1}\Xi_{2}}{\inf\si(H_{\La}(A,V))-\la_{0}},\quad\text{if}\quad\Xi_{1}^{-1}\Xi_{2}\geq-\la_{0}.
\end{aligned} \right.
\end{equation*}
Moreover, the second inequality in \eqref{condition-on-a-1} and the first inequality in
\eqref{condition-on-a-2} correspond to $\Xi_{1}^{-1}\Xi_{2}\leq-\la_{0}$ and
$\Xi_{1}^{-1}\Xi_{2}\geq-\la_{0}$, respectively.

\item[(ii)] The second reason is that \eqref{condition-on-a-2} provides a nonzero lower bound for $a_{0}$, and in turn, an upper bound for $e^{-a_{0}|\beta-\ga|}$, which is important in Section \ref{section-operator-kernel-estimate} because we need such an upper bound, of course after being simplified, to estimate some integrals there.
\end{itemize}
\end{rem}

\begin{lem}\label{lemma-invertibility-of-auxiliary-operator}
Let $A\in\HH_{loc}(\R^{d})$, $V\in\KK_{\pm}(\R^{d})$ and $\La\subset\R^{d}$ open. Let $z\in\rho(H_{\La}(A,V))$, the resolvent set of $H_{\La}(A,V)$. Suppose that $s>0$ and $a\in\R^{d}$ satisfying $|a|=a_{0}>0$ obey \eqref{condition-on-a-1} or \eqref{condition-on-a-2}. Then $H_{\La}^{a}(A,V)-z$ is invertible, i.e., $z\in\rho(H_{\La}^{a}(A,V))$, the resolvent set of $H_{\La}^{a}(A,V)$. In other words, $\rho(H_{\La}(A,V))\subset\rho(H_{\La}^{a}(A,V))$.
\end{lem}
\begin{proof}
By \eqref{expected-equality}, we have
\begin{equation}\label{an-equality}
\begin{split}
H_{\La}^{a}(A,V)-z&=H_{\La}(A,V)-z+\sqrt{\tilde{H}_{\La}(A,V)}B\sqrt{\tilde{H}_{\La}(A,V)}\\
&=(H_{\La}(A,V)-\la_{0})^{\frac{1}{2}}(U+V)(H_{\La}(A,V)-\la_{0})^{\frac{1}{2}},
\end{split}
\end{equation}
where
\begin{equation*}
\begin{split}
U&=(H_{\La}(A,V)-\la_{0})^{-\frac{1}{2}}(H_{\La}(A,V)-z)(H_{\La}(A,V)-\la_{0})^{-\frac{1}{2}}\\
&=(H_{\La}(A,V)-z)(H_{\La}(A,V)-\la_{0})^{-1}\\
\end{split}
\end{equation*}
and
\begin{equation*}
V=(H_{\La}(A,V)-\la_{0})^{-\frac{1}{2}}\sqrt{\tilde{H}_{\La}(A,V)}B\sqrt{\tilde{H}_{\La}(A,V)}(H_{\La}(A,V)-\la_{0})^{-\frac{1}{2}}.
\end{equation*}
Since $(H_{\La}(A,V)-\la_{0})^{\frac{1}{2}}$ is invertible, invertibility of $H_{\La}^{a}(A,V)-z$ is equivalent to that of $U+V$.

We claim that $U+V$ is invertible under the assumption of the lemma with
\begin{equation}\label{bound-of-inverse-of-u+v}
\|(U+V)^{-1}\|\leq
\left\{
\begin{aligned}
\frac{c_{z,\la_{0}}(1-\Th_{1})}{1-\Th_{1}-4sc_{z,\la_{0}}},\quad\text{if}\,\,
a_{0}\,\,\text{satisfies}\,\,\eqref{condition-on-a-1},\\
\frac{(\de-1)\la_{0}c_{z,\la_{0}}}{(\de-1)\la_{0}-2(\de\la_{0}\Xi_{1}+\Xi_{2})c_{z,\la_{0}}},\quad\text{if}\,\,a_{0}\,\,\text{satisfies}\,\,\eqref{condition-on-a-2}.
\end{aligned}
\right.
\end{equation}

Obviously, $U$ is bounded and invertible with $U^{-1}=(H_{\La}(A,V)-\la_{0})(H_{\La}(A,V)-z)^{-1}$. Recall that $\tilde{H}_{\La}(A,V)=H_{\La}(A,V)+\Xi_{1}^{-1}\Xi_{2}\geq0$. Since $\sqrt{\frac{\la+\Xi_{1}^{-1}\Xi_{2}}{\la-\la_{0}}}$ (as a function of $\la$) is bounded on $\si(H_{\La}(A,V))$, both $(H_{\La}(A,V)-\la_{0})^{-\frac{1}{2}}\sqrt{\tilde{H}_{\La}(A,V)}$ and $\sqrt{\tilde{H}_{\La}(A,V)}(H_{\La}(A,V)-\la_{0})^{-\frac{1}{2}}$ are bounded, which implies that $V$ is bounded. Then, by stability of bounded invertibility (see \cite[Theorem IV.1.16]{K76}), it suffices
to require that $\|V\|\|U^{-1}\|<1$. In which case, $U+V$ is invertible with
\begin{equation}\label{bound-of-inverse}
\|(U+V)^{-1}\|\leq\frac{\|U^{-1}\|}{1-\|V\|\|U^{-1}\|}.
\end{equation}
Since $\|U^{-1}\|\leq c_{z,\la_{0}}$ and
\begin{equation*}
\|V\|\leq\|B\|\bigg\|\sqrt{\frac{\la+\Xi_{1}^{-1}\Xi_{2}}{\la-\la_{0}}}\bigg\|^{2}_{L^{\infty}(\si(H_{\La}(A,V)))}\leq2\Xi_{1}\bigg\|\frac{\la+\Xi_{1}^{-1}\Xi_{2}}{\la-\la_{0}}\bigg\|_{L^{\infty}(\si(H_{\La}(A,V)))},
\end{equation*}
it suffices to require
\begin{equation}\label{expected-condition}
2\Xi_{1}c_{z,\la_{0}}\bigg\|\frac{\la+\Xi_{1}^{-1}\Xi_{2}}{\la-\la_{0}}\bigg\|_{L^{\infty}(\si(H_{\La}(A,V)))}<1.
\end{equation}
It is justified in Lemma \ref{condition-invertibility} below that if $s$ and $a$ are as in the assumptions of the current lemma, then \eqref{expected-condition} holds, which then implies that
$U+V$, and hence $H_{\Lambda}^{a}(A,V)-z$, is invertible. Finally, \eqref{bound-of-inverse-of-u+v} follows from \eqref{bound-of-inverse} and Lemma \ref{condition-invertibility} below.
\end{proof}

To finish the proof of Lemma \ref{lemma-invertibility-of-auxiliary-operator}, we show

\begin{lem}\label{condition-invertibility}
Let $z\in\rho(H_{\La}(A,V))$. If $s>0$ and $a_{0}>0$ satisfy \eqref{condition-on-a-1} or \eqref{condition-on-a-2}, then \eqref{expected-condition} holds. Moreover, if \eqref{condition-on-a-1} is satisfied, then $\big\|\frac{\la+\Xi_{1}^{-1}\Xi_{2}}{\la-\la_{0}}\big\|_{L^{\infty}(\si(H_{\La}(A,V)))}=1$ and if \eqref{condition-on-a-2} is satisfied, then $\big\|\frac{\la+\Xi_{1}^{-1}\Xi_{2}}{\la-\la_{0}}\big\|_{L^{\infty}(\si(H_{\La}(A,V)))}\leq\frac{\de\la_{0}+\Xi_{1}^{-1}\Xi_{2}}{(\de-1)\la_{0}}$ for some $\de=\de(\la_{0})\in(0,1)$ satisfying $\de\la_{0}\in\big(\la_{0},\min\{-\Th_{2}, E_{0}\}\big)$.
\end{lem}
\begin{proof}
Instead of proving \eqref{expected-condition} directly, we show how to derive \eqref{condition-on-a-1} or \eqref{condition-on-a-2} so that  \eqref{expected-condition} holds. Recall $\la_{0}<\min\{-\Th_{2},E_{0}\}$, $\Xi_{1}=\frac{2s}{1-\Th_{1}}$, $\Xi_{2}=\frac{2s\Th_{2}}{1-\Th_{1}}+\big(\frac{1}{2s}+\frac{s}{4}\big)a_{0}^{2}$
and $c_{z,\la_{0}}=\big\|\frac{\la-\la_{0}}{\la-z}\big\|_{L^{\infty}(\si(H_{\La}(A,V)))}$. We here discuss two classes of conditions separated by $\Xi_{1}^{-1}\Xi_{2}=-\la_{0}$.

$\rm(i)$ Due to the fact that $\si(H_{\La}(A,V))$ contains a sequence tending to infinity, there holds $\big\|\frac{\la+\Xi_{1}^{-1}\Xi_{2}}{\la-\la_{0}}\big\|_{L^{\infty}(\si(H_{\La}(A,V)))}\geq1$. So the best choice is $\big\|\frac{\la+\Xi_{1}^{-1}\Xi_{2}}{\la-\la_{0}}\big\|_{L^{\infty}(\si(H_{\La}(A,V)))}=1$, which holds if and only if $\Xi_{1}^{-1}\Xi_{2}\leq-\la_{0}$
since $\inf\si(H_{\La}(A,V))+\Xi_{1}^{-1}\Xi_{2}\geq0$. By making $a_{0}$ small enough, the condition $\Xi_{1}^{-1}\Xi_{2}\leq-\la_{0}$ is readily satisfied. Thus \eqref{expected-condition} reduces to
\begin{equation}\label{app-expected-condition-reduced}
2\Xi_{1}c_{z,\la_{0}}<1.
\end{equation}
Note $\lim_{\la\ra\infty}\big|\frac{\la-\la_{0}}{\la-z}\big|=1$ pointwise in $z\in\rho(H_{\La}(A,V))$ and $\la_{0}<\min\{-\Th_{2},E_{0}\}$, which implies that $c_{z,\la_{0}}\geq1$. Hence, if \eqref{app-expected-condition-reduced} holds, then automatically, $\Xi_{1}<\frac{1}{2}<1$.

For any fixed $z\in\rho(H_{\La}(A,V))$ and $\la_{0}<\min\{-\Th_{2},E_{0}\}$, there exists $s$ such that \eqref{app-expected-condition-reduced} is satisfied. Moreover, $s$ can not be chosen to be independent of $z$ or $\la_{0}$ because of the facts that
\begin{equation*}
\lim_{\substack{z\in\rho(H_{\La}(A,V))\\\text{dist}(z,\si(H_{\La(A,V)))\ra0}}}c_{z,\la_{0}}=\infty\quad\text{pointwise in}\,\,\la_{0}<\min\{-\Th_{2},E_{0}\}.
\end{equation*}
or
\begin{equation*}
\lim_{\la_{0}\ra-\infty}c_{z,\la_{0}}=\infty\quad\text{pointwise in}\,\,z\in\rho(H_{\La}(A,V)),
\end{equation*}
respectively.

Explicitly, we can choose $s$ to be any number satisfying
\begin{equation}\label{condition-on-s}
s\in\bigg(0,\frac{1-\Th_{1}}{4c_{z,\la_{0}}}\bigg)\subset\bigg(0,\frac{1-\Th_{1}}{2}\bigg)
\end{equation}
so that \eqref{app-expected-condition-reduced} is satisfied, so $\Xi_{1}<1$ holds. Then, by requiring $a_{0}$ to satisfy
\begin{equation}\label{condition-on-a}
a_{0}^{2}\leq\frac{2s(\la_{0}+\Th_{2})}{\Theta_{1}-1}\bigg(\frac{1}{2s}+\frac{s}{4}\bigg)^{-1},
\end{equation}
the condition $\Xi_{1}^{-1}\Xi_{2}\leq-\la_{0}$ holds. Consequently, any pair $(s,a_{0})$ satisfying \eqref{condition-on-s} and \eqref{condition-on-a} guarantees \eqref{expected-condition}.

$\rm(ii)$ Now, we require $\Xi_{1}^{-1}\Xi_{2}\geq-\la_{0}$. Then,
$\frac{\la+\Xi_{1}^{-1}\Xi_{2}}{\la-\la_{0}}$, as a function of $\la$, is decreasing on $(\la_{0},\infty)$,
which implies
\begin{equation*}
\bigg\|\frac{\la+\Xi_{1}^{-1}\Xi_{2}}{\la-\la_{0}}\bigg\|_{L^{\infty}(\si(H_{\La}(A,V)))}\leq\frac{\la_{*}+\Xi_{1}^{-1}\Xi_{2}}{\la_{*}-\la_{0}},\quad\forall\la_{*}\in\big(\la_{0},\min\{-\Th_{2},E_{0}\}\big).
\end{equation*}
In particular, we can take $\la_{*}=\de\la_{0}$ for some $\de=\de(\la_{0})\in(0,1)$ and obtain
\begin{equation*}
\bigg\|\frac{\la+\Xi_{1}^{-1}\Xi_{2}}{\la-\la_{0}}\bigg\|_{L^{\infty}(\si(H_{\La}(A,V)))}\leq\frac{\de\la_{0}+\Xi_{1}^{-1}\Xi_{2}}{(\de-1)\la_{0}}.
\end{equation*}
Then, $2\Xi_{1}c_{z,\la_{0}}\frac{\de\la_{0}+\Xi_{1}^{-1}\Xi_{2}}{(\de-1)\la_{0}}<1$, i.e, $\Xi_{2}<\frac{(\de-1)\la_{0}}{2c_{z,\la_{0}}}-\de\la_{0}\Xi_{1}$, will ensure \eqref{expected-condition}. Moreover, considering the assumption $\Xi_{1}^{-1}\Xi_{2}\geq-\lambda_{0}$, we deduce $2\Xi_{1}c_{z,\la_{0}}<1$, which leads to $\Xi_{1}<\frac{1}{2}$ since $c_{z,\la_{0}}\geq1$. In conclusion, to ensure \eqref{expected-condition}, we only need to require
\begin{equation*}
-\la_{0}\Xi_{1}\leq\Xi_{2}<\frac{(\de-1)\la_{0}}{2c_{z,\la_{0}}}-\de\la_{0}\Xi_{1}.
\end{equation*}

Explicitly, if $s>0$ and $a_{0}>0$ satisfy
\begin{equation*}
\frac{2s(\la_{0}+\Th_{2})}{\Th_{1}-1}\bigg(\frac{1}{2s}+\frac{s}{4}\bigg)^{-1}\leq a_{0}^{2}<\bigg(\frac{(\de-1)\la_{0}}{2c_{z,\la_{0}}}+\frac{2s(\de\la_{0}+\Th_{2})}{\Th_{1}-1}\bigg)\bigg(\frac{1}{2s}+\frac{s}{4}\bigg)^{-1}
\end{equation*}
for some  $\de=\de(\la_{0},\Th_{2},E_{0})\in(0,1)$ such that $\de\la_{0}\in\big(\la_{0},\min\{-\Th_{2},E_{0}\}\big)$, then $\Xi_{1}<1$ and \eqref{expected-condition} holds.
\end{proof}

We proceed to prove Theorem \ref{main-theorem}. Since the proof in the case $n\geq2$ is based on the proof in the case $n=1$, we divide Theorem \ref{main-theorem} into two parts according to $n=1$ and $n\geq2$. Moreover, we restate the theorem in the cases $n=1$ and $n\geq2$ as Theorem \ref{theorem-decay-estimate} and Theorem \ref{corollary-decay-estimate} below, respectively. For notational simplicity, we set
\begin{equation}\label{bound-of-inverse-of-u+v-const}
C_{*}=
\left\{
\begin{aligned}
\frac{c_{z,\la_{0}}(1-\Th_{1})}{1-\Th_{1}-4sc_{z,\la_{0}}},\quad\text{if}\,\,
a_{0}\,\,\text{satisfies}\,\,\eqref{condition-on-a-1},\\
\frac{(\de-1)\la_{0}c_{z,\la_{0}}}{(\de-1)\la_{0}-2(\de\la_{0}\Xi_{1}+\Xi_{2})c_{z,\la_{0}}},\quad\text{if}\,\,a_{0}\,\,\text{satisfies}\,\,\eqref{condition-on-a-2}.
\end{aligned}
\right.
\end{equation}
Then, $\|(U+V)^{-1}\|\leq C_{*}$.

\begin{thm}\label{theorem-decay-estimate}
Let $A\in\mathcal{H}_{loc}(\R^{d})$, $V\in\mathcal{K}_{\pm}(\R^{d})$ and $\La\subset\R^{d}$ open. Suppose $p>\frac{d}{2}$. Let $z\in\rho(H_{\La}(A,V))$, the resolvent set of $H_{\La}(A,V)$. Suppose that $s>0$ and $a_{0}>0$ satisfy \eqref{condition-on-a-1} or \eqref{condition-on-a-2}.  Then, for any $\beta,\ga\in\R^{d}$,
\begin{equation}\label{expected-estimate}
\|\chi_{\beta}(H_{\La}(A,V)-z)^{-1}\chi_{\ga}\|_{\mathcal{J}_{p}}\leq C_{p,\la_{0}}C_{*}e^{\sqrt{d}a_{0}}e^{-a_{0}|\beta-\ga|},
\end{equation}
where $C_{p,\la_{0}}>0$ depends only on $p$ and $\la_{0}$.
\end{thm}
\begin{proof}
By Lemma \ref{lemma-invertibility-of-auxiliary-operator} and the operator equality \eqref{an-equality}, we have
\begin{equation*}
\chi_{\beta}(H_{\La}^{a}(A,V)-z)^{-1}\chi_{\ga}=\chi_{\beta}(H_{\La}(A,V)-\la_{0})^{-\frac{1}{2}}(U+V)^{-1}(H_{\La}(A,V)-\la_{0})^{-\frac{1}{2}}\chi_{\ga}.
\end{equation*}
Since the function $(\la-\la_{0})^{-\frac{1}{2}}$ satisfies \eqref{assumption-f} with $\alpha=\frac{1}{2}$, $\frac{1}{2}>\frac{d}{2\cdot2p}$ and $2p>d\geq2$, Theorem
\ref{theorem-trace-ideal-estimate} insures that both $\chi_{\beta}(H_{\La}(A,V)-\la_{0})^{-\frac{1}{2}}$ and $(H_{\La}(A,V)-\la_{0})^{-\frac{1}{2}}\chi_{\ga}$ are in $\JJ_{2p}$ with $\mathcal{J}_{2p}$-norm only depending on $p$ and $\la_{0}$. It then follows that $\chi_{\beta}(H_{\La}^{a}(A,V)-z)^{-1}\chi_{\ga}\in\JJ_{p}$
with
\begin{equation*}
\begin{split}
&\|\chi_{\beta}(H_{\La}^{a}(A,V)-z)^{-1}\chi_{\ga}\|_{\JJ_{p}}\\
&\quad\quad\leq\|\chi_{\beta}(H_{\La}(A,V)-\la_{0})^{-\frac{1}{2}}\|_{\JJ_{2p}}\|(U+V)^{-1}\|\cdot\|(H_{\La}(A,V)-\la_{0})^{-\frac{1}{2}}\chi_{\ga}\|_{\JJ_{2p}}\\
&\quad\quad\leq C_{p,\la_{0}}C_{*}
\end{split}
\end{equation*}
where \eqref{bound-of-inverse-of-u+v} and \eqref{bound-of-inverse-of-u+v-const} are used and $C_{p,\la_{0}}>0$ only depends on $p$ and $\la_{0}$. Considering \eqref{auxiliary-operator}, we obtain
\begin{equation*}
\begin{split}
&\|\chi_{\beta}(H_{\La}(A,V)-z)^{-1}\chi_{\ga}\|_{\JJ_{p}}\\
&\quad\quad=\|\chi_{\beta}e^{-a\cdot x}(H_{\La}^{a}(A,V)-z)^{-1}e^{a\cdot
x}\chi_{\ga}\|_{\JJ_{p}}\\
&\quad\quad=\|e^{-a\cdot(\beta-\ga)}\big(e^{-a\cdot(x-\beta)}\chi_{\beta}\big)\big(\chi_{\beta}(H_{\La}^{a}(A,V)-z)^{-1}\chi_{\ga}\big)\big(\chi_{\ga}e^{a\cdot(x-\ga)}\big)\|_{\JJ_{p}}\\
&\quad\quad\leq\|\chi_{\beta}(H_{\La}^{a}(A,V)-z)^{-1}\chi_{\ga}\|_{\mathcal{J}_{p}}\|e^{-a\cdot(x-\beta)}\chi_{\beta}\|\cdot\|\chi_{\ga}e^{a\cdot(x-\ga)}\|e^{-a\cdot(\beta-\ga)}\\
&\quad\quad\leq C_{p,\la_{0}}C_{*}e^{\sqrt{d}a_{0}}e^{-a_{0}|\beta-\ga|}
\end{split}
\end{equation*}
where we used the fact that both $\|e^{-a\cdot(x-\beta)}\chi_{\beta}\|$ and $\|\chi_{\ga}e^{a\cdot(x-\ga)}\|$ are bounded by $e^{\frac{\sqrt{d}}{2}|a|}$. By choosing $a=a_{0}|\beta-\ga|^{-1}(\beta-\ga)$, we find \eqref{expected-estimate}.
\end{proof}

\begin{thm}\label{corollary-decay-estimate}
Let $A\in\HH_{loc}(\R^{d})$, $V\in\KK_{\pm}(\R^{d})$ and $\La\subset\R^{d}$ open. Suppose $p>\frac{d}{2n}$ with $n\in\N$ and $n\geq2$. Let $z\in\rho(H_{\La}(A,V))$. Suppose that $s>0$ and $a_{0}>0$ satisfy \eqref{condition-on-a-1} or \eqref{condition-on-a-2}. Then, for any $\de_{0}\in(0,1)$ and any $\beta,\ga\in\R^{d}$, there holds
\begin{equation*}
\|\chi_{\beta}(H_{\La}(A,V)-z)^{-n}\chi_{\ga}\|_{\JJ_{p}}\leq(C_{p,n,\la_{0}}c_{\de_{0},a_{0}}C_{*})^{n-1}e^{(n-1)\sqrt{d}a_{0}}e^{-\de_{0}a_{0}|\beta-\ga|},
\end{equation*}
where $C_{p,n,\la_{0}}>0$ only depends on $p$, $n$ and $\la_{0}$ and $c_{\de_{0},a_{0}}=\sum_{\al\in\Z^{d}}e^{-(1-\de_{0})a_{0}|\al|}<\infty$.
\end{thm}
\begin{proof}
Write
\begin{equation*}
\chi_{\beta}(H_{\La}(A,V)-z)^{-n}\chi_{\ga}=\sum_{\substack{\al_{j}\in\Z^{d}\\j=1,\dots,n-1}}R_{\beta,\al_1}R_{\al_{1},\al_2}\cdots
R_{\al_{n-2},\al_{n-1}}R_{\al_{n-1},\ga},
\end{equation*}
where
$R_{\beta,\al_1}=\chi_{\beta}(H_{\La}(A,V)-z)^{-1}\chi_{\al_{1}}$,
$R_{\al_{j},\al_{j+1}}=\chi_{\al_{j}}(H_{\La}(A,V)-z)^{-1}\chi_{\al_{j+1}}$, $j=1,\dots,n-2$ and $R_{\al_{n-1},\ga}=\chi_{\al_{n-1}}(H_{\La}(A,V)-z)^{-1}\chi_{\ga}$. Since $pn>\frac{d}{2}$ by assumption, Theorem \ref{theorem-decay-estimate} says that
\begin{equation*}
\|\chi_{x}(H_{\La}(A,V)-z)^{-1}\chi_{y}\|_{\JJ_{pn}}\leq C_{p,n,\la_{0}}C_{*}e^{\sqrt{d}a_{0}}e^{-a_{0}|\beta-\ga|},\quad\forall x,y\in\mathbb{R}^{d},
\end{equation*}
where $C_{p,n,\la_{0}}>0$ only depends on $p$, $n$ and $\la_{0}$. By H\"{o}lder's inequality for trace ideals (see \cite[Theorem 2.8]{Si05}), the result of the corollary follows from
\begin{equation}\label{inequality-1}
\begin{split}
&\sum_{\substack{\al_{j}\in\Z^{d}\\j=1,\dots,n-1}}e^{-a_{0}|\beta-\al_{1}|}e^{-a_{0}|\al_{1}-\al_{2}|}\cdots
e^{-a_{0}|\al_{n-2}-\al_{n-1}|}e^{-a_{0}|\al_{n-1}-\beta|}\\
&\quad\quad\leq c_{\de_{0},a_{0}}^{n-1}e^{-\de_{0} a_{0}|\beta-\ga|},\quad\forall\de_{0}\in(0,1),
\end{split}
\end{equation}
where
$c_{\de_{0},a_{0}}=\sum_{\al\in\Z^{d}}e^{-(1-\de_{0})a_{0}|\al|}<\infty$.

To complete the proof, we show \eqref{inequality-1}. Pick and fix any $\de_{0}\in(0,1)$. First, we have from the triangular inequality and Cauchy's inequality
\begin{equation*}
\begin{split}
&\sum_{\al_{1}\in\Z^{d}}e^{-a_{0}|\beta-\al_{1}|}e^{-a_{0}|\al_{1}-\al_{2}|}\\
&\quad\quad=\sum_{\al_{1}\in\Z^{d}}e^{-(1-\de_{0})a_{0}|\beta-\al_{1}|}e^{-\de_{0}
a_{0}(|\beta-\al_{1}|+|\al_{1}-\al_{2}|)}e^{-(1-\de_{0})a_{0}|\al_{1}-\al_{2}|}\\
&\quad\quad\leq e^{-\de_{0}
a_{0}|\beta-\al_{2}|}\sum_{\al_{1}\in\Z^{d}}e^{-(1-\de_{0})a_{0}|\beta-\al_{1}|}e^{-(1-\de_{0})a_{0}|\al_{1}-\al_{2}|}\\
&\quad\quad\leq e^{-\de_{0}a_{0}|\beta-\al_{2}|}\bigg(\sum_{\al_{1}\in\Z^{d}}e^{-2(1-\de_{0})a_{0}|\beta-\al_{1}|}\bigg)^{\frac{1}{2}}\bigg(\sum_{\al_{1}\in\Z^{d}}e^{-2(1-\de_{0})a_{0}|\al_{1}-\al_{2}|}\bigg)^{\frac{1}{2}}\\
&\quad\quad\leq c_{\de_{0},a_{0}}e^{-\de_{0}a_{0}|\beta-\al_{2}|},
\end{split}
\end{equation*}
where
$c_{\de_{0},a_{0}}=\sum_{\al_{1}\in\Z^{d}}e^{-(1-\de_{0})a_{0}|\al_{1}|}\geq\sum_{\al_{1}\in\Z^{d}}e^{-2(1-\de_{0})a_{0}|\al_{1}|}$.
Next, by the above estimate and the triangular inequality,
\begin{equation*}
\begin{split}
&\sum_{\al_{2}\in\Z^{d}}\sum_{\al_{1}\in\Z^{d}}e^{-a_{0}|\beta-\al_{1}|}e^{-a_{0}|\al_{1}-\al_{2}|}e^{-a_{0}|\al_{2}-\al_{3}|}\\
&\quad\quad\leq
c_{\de,a_{0}}\sum_{\al_{2}\in\Z^{d}}e^{-\de a_{0}|\beta-\al_{2}|}e^{-a_{0}|\al_{2}-\al_{3}|}\\
&\quad\quad=c_{\de_{0},a_{0}}\sum_{\al_{2}\in\Z^{d}}e^{-\de_{0}
a_{0}|\beta-\al_{2}|}e^{-\de_{0}a_{0}|\al_{2}-\al_{3}|}e^{-(1-\de_{0})a_{0}|\al_{2}-\al_{3}|}\\
&\quad\quad\leq c_{\de,a_{0}}e^{-\de_{0}a_{0}|\beta-\al_{3}|}\sum_{\al_{2}\in\Z^{d}}e^{-(1-\de_{0})a_{0}|\al_{2}-\al_{3}|}\\
&\quad\quad=c_{\de_{0},a_{0}}^{2}e^{-\de_{0}a_{0}|\beta-\al_{3}|}.
\end{split}
\end{equation*}
By induction, we find \eqref{inequality-1}. This completes the proof.
\end{proof}


\section{The Operator Kernel Estimate in Trace
Ideals}\label{section-operator-kernel-estimate}

In this section, we study the operator kernel estimate in trace-class norms. More precisely, we prove polynomial decay, in trace ideals, of operators
\begin{equation*}
\chi_{\beta}f(H_{\La}(A,V))\chi_{\ga},\quad\beta,\ga\in\R^{d}
\end{equation*}
in terms of $|\beta-\ga|$ for $f\in\SS(\R)$, the Schwartz space. The main result in this section is stated in Theorem \ref{main-theorem-poly}, whose proof is based on Theorem \ref{main-theorem} (in fact, Theorem \ref{theorem-decay-estimate}) and the Helffer-Sj\"{o}strand formula (see \cite{HS89}), which is
defined for a much larger class of slowly decreasing smooth functions on $\R$, denoted by $\mathscr{A}$. See Appendix \ref{app-Helffer-Sjostrand-formula} for the definition of $\mathscr{A}$ and the Helffer-Sj\"{o}strand formula.

Before proving Theorem \ref{main-theorem-poly}, we first simplify the second estimate in \eqref{expected-estimate} by adding more conditions so that this estimate can by easily used. Our idea is as follows: by the Helffer-Sj\"{o}strand formula \eqref{the-Helffer-Sjostrand-formula}, we have for any $f\in\SS(\R)$,
\begin{equation*}
\chi_{\beta}f(H_{\La}(A,V))\chi_{\ga}=\frac{1}{\pi}\int_{\R^{2}}\frac{\pa\tilde{f}_{n}(z)}{\pa\bar{z}}\chi_{\beta}(H_{\La}(A,V)-z)^{-1}\chi_{\ga}dudv,\quad\forall\,\,n\geq1.
\end{equation*}
Therefore, by \eqref{app-an-inequality},
\begin{equation}\label{estimate-integrals}
\begin{split}
&\|\chi_{\beta}f(H_{\La}(A,V))\chi_{\ga}\|_{\JJ_{p}}\\
&\quad\quad\leq\frac{C}{\pi}\sum_{r=0}^{n}\frac{1}{r!}\int_{U}|f^{(r)}(u)|\frac{|v|^{r}}{\lan u\ran}\|\chi_{\beta}(H_{\La}(A,V)-z)^{-1}\chi_{\ga}\|_{\JJ_{p}}dudv\\
&\quad\quad\quad+\frac{1}{2\pi n!}\int_{V}|f^{(n+1)}(u)||v|^{n}\|\chi_{\beta}(H_{\La}(A,V)-z)^{-1}\chi_{\ga}\|_{\JJ_{p}}dudv
\end{split}
\end{equation}
for any $n\geq1$. Clearly, in order to estimate the integrals on the right hand side of \eqref{estimate-integrals}, we need \eqref{expected-estimate}. More precisely, we need
\begin{equation}\label{expected-estimate-2ed}
\|\chi_{\beta}(H_{\La}(A,V)-z)^{-1}\chi_{\ga}\|_{\mathcal{J}_{p}}\leq\frac{C_{p,\la_{0}}(\de-1)\la_{0}c_{z,\la_{0}}}{(\de-1)\la_{0}-2(\de\la_{0}\Xi_{1}+\Xi_{2})c_{z,\la_{0}}}e^{\sqrt{d}a_{0}}e^{-a_{0}|\beta-\ga|},
\end{equation}
since the conditions ensuring it provide a nonzero lower bound for $a_{0}$, which in turn provide an upper bound for the exponential term. However, this estimate is too rough to deal with since many parameters in the upper bound depend on $z$. To simplify it, we put more conditions on $s$ and $a_{0}$.

For $s>0$, we assume
\begin{equation}\label{condition-on-s-additional}
s<\frac{1}{2}\frac{1-\Th_{1}}{4c_{z,\la_{0}}}\frac{1-\de}{2-\de}
\end{equation}
such that
\begin{equation}\label{condition-}
\frac{2s(2\la_{0}+\Th_{2})}{\Th_{1}-1}<\frac{(\de-1)\la_{0}}{4c_{z,\la_{0}}}+\frac{2s(\de\la_{0}+\Th_{2})}{\Th_{1}-1}<\frac{(\de-1)\la_{0}}{2c_{z,\la_{0}}}+\frac{2s(\de\la_{0}+\Th_{2})}{\Th_{1}-1}.
\end{equation}
For $a_{0}>0$, we require
\begin{equation}\label{condition-on-a-additional}
\frac{2s(\la_{0}+\Th_{2})}{\Th_{1}-1}\bigg(\frac{1}{2s}+\frac{s}{4}\bigg)^{-1}\leq a_{0}^{2}\leq\frac{2s(2\la_{0}+\Th_{2})}{\Th_{1}-1}\bigg(\frac{1}{2s}+\frac{s}{4}\bigg)^{-1}.
\end{equation}
The intuitive interpretation of these conditions is that we make $2\Xi_{1}c_{z,\la_{0}}$ smaller and bound $\Xi_{1}^{-1}\Xi_{2}$ by $-2\la_{0}$ from above. Indeed, \eqref{condition-on-s-additional} is equivalent to
\begin{equation*}
2\Xi_{1}c_{z,\la_{0}}<\frac{1}{2}\frac{1-\de}{2-\de}
\end{equation*}
and, \eqref{condition-} and \eqref{condition-on-a-additional} are equivalent to
\begin{equation*}
-\la_{0}\Xi_{1}\leq\Xi_{2}\leq-2\la_{0}\Xi_{1}<\frac{(\de-1)\la_{0}}{4c_{z,\la_{0}}}-\de\la_{0}\Xi_{1}<\frac{(\de-1)\la_{0}}{2c_{z,\la_{0}}}-\de\la_{0}\Xi_{1}.
\end{equation*}
Clearly, \eqref{condition-on-s-additional} and \eqref{condition-on-a-additional} are stronger than \eqref{condition-on-a-2}. Hence, under the assumptions of \eqref{condition-on-s-additional} and \eqref{condition-on-a-additional}, \eqref{expected-estimate-2ed} holds. Moreover,
\begin{equation*}
\frac{(\de-1)\la_{0}c_{z,\la_{0}}}{(\de-1)\la_{0}-2(\de\la_{0}\Xi_{1}+\Xi_{2})c_{z,\la_{0}}}=\frac{c_{z,\la_{0}}}{1-2\Xi_{1}c_{z,\la_{0}}\frac{\de\la_{0}+\Xi_{1}^{-1}\Xi_{2}}{(\de-1)\la_{0}}}\leq2c_{z,\la_{0}}
\end{equation*}
and, therefore, \eqref{expected-estimate-2ed} can be simplified to
\begin{equation}\label{expected-estimate-simplied}
\|\chi_{\beta}(H_{\La}(A,V)-z)^{-1}\chi_{\ga}\|_{\JJ_{p}}\leq
C_{p,\la_{0}}c_{z,\la_{0}}e^{\sqrt{d}a_{0}}e^{-a_{0}|\beta-\ga|},
\end{equation}
where $C_{p,\la_{0}}>0$ depends only on $p$ and $\la_{0}$. Further, we rewrite \eqref{condition-on-a-additional} as
\begin{equation}\label{condition-on-a-additional-simplified}
\frac{2(\la_{0}+\Th_{2})}{\Th_{1}-1}\bigg(\frac{1}{2s^{2}}+\frac{1}{4}\bigg)^{-1}\leq a_{0}^{2}\leq\frac{2(2\la_{0}+\Th_{2})}{\Th_{1}-1}\bigg(\frac{1}{2s^{2}}+\frac{1}{4}\bigg)^{-1},
\end{equation}
which in particular says that $a_{0}$ is bounded from above by a constant independent of $z$, which implies that $e^{\sqrt{d}a_{0}}$ is bounded from above by a constant independent of $z$. Hence, \eqref{expected-estimate-simplied} is further simplified to
\begin{equation}\label{expected-estimate-further-simplied}
\|\chi_{\beta}(H_{\La}(A,V)-z)^{-1}\chi_{\ga}\|_{\JJ_{p}}\leq
C_{p,\la_{0}}c_{z,\la_{0}}e^{-a_{0}|\beta-\ga|},
\end{equation}
where $C_{p,\la_{0}}>0$ only depends on $p$ and $\la_{0}$.

Note that the lower bound of $a_{0}$ is not very easy to handle because of the uncertainty of $s$ and the quantity $c_{z,\la_{0}}$. To find a simpler lower bound for $a_{0}$, we first fix some $s$, say $s=\frac{1}{4}\frac{1-\Th_{1}}{4c_{z,\la_{0}}}\frac{1-\de}{2-\de}$, and then give explicit bound for $c_{z,\la_{0}}$ with $z=u+iv$ under assumptions $(u,v)\in U$ and $(u,v)\in V$, respectively. Recall
\begin{equation*}
\begin{split}
&c_{z,\la_{0}}=\bigg\|\frac{\la-\la_{0}}{\la-z}\bigg\|_{L^{\infty}(\si(H_{\La}(A,V)))},\\
&U=\big\{(u,v)\in\R^{2}|\lan u\ran<|v|<2\lan u\ran\big\}\\
\text{and}\,\,&V=\big\{(u,v)\in\R^{2}|0<|v|<2\lan u\ran\big\}.
\end{split}
\end{equation*}

\begin{lem}
Let $z\in\rho(H_{\La}(A,V))$. Then, with $z=u+iv$,
\begin{equation}\label{expected-estimate-simpliest-form}
\begin{split}
\|\chi_{\beta}(H_{\La}(A,V)-z)^{-1}\chi_{\ga}\|_{\JJ_{p}}\leq
\left\{
\begin{aligned}
C_{p,\la_{0}}|v|e^{-\frac{C_{\la_{0}}}{|v|}|\beta-\ga|},\quad\text{if}\,\,(u,v)\in U,\\
C_{p,\la_{0}}\frac{\lan u\ran}{|v|}e^{-C_{\la_{0}}\frac{|v|}{\lan u\ran}|\beta-\ga|},\quad\text{if}\,\,(u,v)\in V,
\end{aligned}
\right.
\end{split}
\end{equation}
where $C_{\la_{0}}>0$ depends only on $\la_{0}$ and $C_{p,\la_{0}}>0$ depends only on $p$ and $\la_{0}$.
\end{lem}
\begin{proof}
For any $z\in\rho(H_{\La}(A,V))$, we let $s=\frac{1}{4}\frac{1-\Th_{1}}{4c_{z,\la_{0}}}\frac{1-\de}{2-\de}$ and $a_{0}>0$ satisfy \eqref{condition-on-a-additional} such that \eqref{expected-estimate-further-simplied} holds. By the first inequality in
\eqref{condition-on-a-additional-simplified} and $s=\frac{1}{4}\frac{1-\Th_{1}}{4c_{z,\la_{0}}}\frac{1-\de}{2-\de}$, we have
\begin{equation}\label{lower-bound-a-0}
\begin{split}
a_{0}^{2}\geq\frac{2(\la_{0}+\Th_{2})}{\Th_{1}-1}\bigg(\frac{1}{2s^{2}}+\frac{1}{4}\bigg)^{-1}=\frac{-(\la_{0}+\Th_{2})\tilde{C}}{C_{\la_{0}}c_{z,\la_{0}}^{2}+\bar{C}}
\end{split}
\end{equation}
for some $\bar{C}>0$ and $\tilde{C}>0$.

Let $(u,v)\in U$. For any $\la\in\si(H_{\La}(A,V))$, $|\la-z|\geq\text{dist}(z,\si(H_{\La}(A,V)))\geq|v|>\lan u\ran\geq1$, which implies that
\begin{equation*}
c_{z,\la_{0}}\leq\bigg\|1+\frac{|z-\la_{0}|}{|\la-z|}\bigg\|_{L^{\infty}(\si(H_{\La}(A,V)))}\leq1+|z|-\la_{0}\leq1-\la_{0}+\sqrt{2}|v|\leq C_{\la_{0}}|v|
\end{equation*}
and then
\begin{equation}\label{lower-bound-a-0-1}
a_{0}\geq\sqrt{\frac{-(\la_{0}+\Th_{2})\tilde{C}}{C_{\la_{0}}c_{z,\la_{0}}^{2}+\bar{C}}}\geq\frac{C_{\la_{0}}}{|v|},
\end{equation}
where the fact $|v|\geq1$ is used.

Let $(u,v)\in V$. Then
\begin{equation*}
c_{z,\la_{0}}\leq\bigg\|1+\frac{|z-\la_{0}|}{|\la-z|}\bigg\|_{L^{\infty}(\si(H_{\La}(A,V)))}\leq1+\frac{|z|-\la_{0}}{|v|}\leq\frac{5\lan u\ran-\la_{0}}{|v|}\leq C_{\la_{0}}\frac{\lan u\ran}{|v|},
\end{equation*}
which, together with \eqref{lower-bound-a-0}, implies that
\begin{equation}\label{lower-bound-a-0-2}
a_{0}\geq C_{\la_{0}}\frac{|v|}{\lan u\ran},
\end{equation}
where the fact $\frac{\lan u\ran}{|v|}>\frac{1}{2}$ is used.

By means of \eqref{expected-estimate-further-simplied}, \eqref{lower-bound-a-0-1} and \eqref{lower-bound-a-0-2}, we find \eqref{expected-estimate-simpliest-form}.
\end{proof}

We now restate and prove Theorem \ref{main-theorem-poly}.

\begin{thm}\label{theorem-poly-decay-of-operator-kernel}
Let $A\in\HH_{loc}(\R^{d})$, $V\in\KK_{\pm}(\R^{d})$ and $\La\subset\R^{d}$ open. Suppose $p>\frac{d}{2}$. Then, for any $f\in\SS(\R)$ and any $k\in\N$,
\begin{equation*}
\|\chi_{\beta}f(H_{\La}(A,V))\chi_{\ga}\|_{\JJ_{p}}\leq C_{p,\la_{0},k,f}|\beta-\ga|^{-k},\quad\forall\,\,\beta,\ga\in\R^{d},
\end{equation*}
where $C_{p,\la_{0},k,f}>0$ depends only on $p$, $\la_{0}$, $k$ and $f$.
\end{thm}

\begin{proof}
Fix any $k\in\N$ and let $n=k+1$ in \eqref{estimate-integrals}. Since the function $\tha(t)=e^{-t}t^{k}$, $t\geq0$ attains its global maximum at $t=k$, we have
\begin{equation}\label{an-algebraic-inequaltiy}
e^{-t}\leq e^{-k}k^{k}t^{-k}.
\end{equation}
Applying \eqref{an-algebraic-inequaltiy} to $t=\frac{C_{\la_{0}}}{|v|}|\beta-\ga|$ and
$t=C_{\la_{0}}\frac{|v|}{\lan u\ran}|\beta-\ga|$, respectively, we obtain
\begin{equation}\label{an-algebraic-inequaltiy-1}
e^{-\frac{C_{\la_{0}}}{|v|}|\beta-\ga|}\leq\frac{e^{-k}k^{k}}{C_{\la_{0}}^{k}|\beta-\ga|^{k}}|v|^{k}
\end{equation}
and
\begin{equation}\label{an-algebraic-inequaltiy-2}
e^{-C_{\la_{0}}\frac{|v|}{\lan
u\ran}|\beta-\ga|}\leq\frac{e^{-k}k^{k}}{C_{\la_{0}}^{k}|\beta-\ga|^{k}}\frac{\lan u\ran^{k}}{|v|^{k}},
\end{equation}
respectively.

We now use \eqref{an-algebraic-inequaltiy-1} and \eqref{an-algebraic-inequaltiy-2} to estimate the integrals in \eqref{estimate-integrals}. By the first estimate in \eqref{expected-estimate-simpliest-form} and \eqref{an-algebraic-inequaltiy-1}, we have for some $C_{p,\la_{0},k,f}>0$,
\begin{equation*}
\begin{split}
&\sum_{r=0}^{k+1}\frac{1}{r!}\int_{U}|f^{(r)}(u)|\frac{|v|^{r}}{\lan
u\ran}\|\chi_{\beta}(H_{\La}(A,V)-z)^{-1}\chi_{\ga}\|_{\JJ_{p}}dudv\\
&\quad\quad\leq\frac{C_{p,\la_{0}}e^{-k}k^{k}}{C_{\la_{0}}^{k}|\beta-\ga|^{k}}\sum_{r=0}^{k+1}\frac{1}{r!}\int_{U}|f^{(r)}(u)|\frac{|v|^{r+k+1}}{\lan
u\ran}dudv\\
&\quad\quad=\frac{C_{p,\la_{0}}e^{-k}k^{k}}{C_{\la_{0}}^{k}|\beta-\ga|^{k}}\sum_{r=0}^{k+1}\frac{1}{r!}\frac{2^{r+k+3}-2}{r+k+2}\int_{\R}|f^{(r)}(u)|\lan u\ran^{r+k+1}du\\
&\quad\quad\leq C_{p,\la_{0},k,f}|\beta-\ga|^{-k},
\end{split}
\end{equation*}
where the fact $f\in\mathcal{S}(\mathbb{R})$, so the integrals are convergent, is used.

Similarly, by the second estimate in \eqref{expected-estimate-simpliest-form} and \eqref{an-algebraic-inequaltiy-2},
\begin{equation*}
\begin{split}
&\frac{1}{2\pi
(k+1)!}\int_{V}|f^{(n+1)}(u)||v|^{n}\|\chi_{\beta}(H_{\La}(A,V)-z)^{-1}\chi_{\ga}\|_{\JJ_{p}}dudv\\
&\quad\quad\leq\frac{C_{p,\la_{0}}e^{-k}k^{k}}{2\pi(k+1)!C_{\la_{0}}^{k}|\beta-\ga|^{k}}\int_{V}|f^{(k+2)}(u)|\lan u\ran^{k+1}dudv\\
&\quad\quad=\frac{4C_{p,\la_{0}}e^{-k}k^{k}}{2\pi(k+1)!C_{\la_{0}}^{k}|\beta-\ga|^{k}}\int_{\R}|f^{(k+2)}(u)|\lan u\ran^{k+2}du\\
&\quad\quad\leq C_{p,\la_{0},k,f}|\beta-\ga|^{-k}.
\end{split}
\end{equation*}

Consequently, for any $f\in\SS(\R)$, there exists $C_{p,\la_{0},k,f}>0$ so that
\begin{equation*}
\|\chi_{\beta}f(H_{\La}(A,V))\chi_{\ga}\|_{\JJ_{p}}\leq
C_{p,\la_{0},k,f}|\beta-\ga|^{-k},\quad\forall\,\,\beta,\ga\in\R^{d}.
\end{equation*}
This proves Theorem \ref{theorem-poly-decay-of-operator-kernel}.
\end{proof}

\section*{Acknowledgments} The author would like to thank Prof. Wenxian Shen for her support during the work of this paper.


\appendix

\section{Sectorial Form and $m$-Sectorial Operator}\label{app-sectorial-form}

In this section, we review some results about sectorial form and $m$-sectorial operator used in the above sections. The material is chosen from \cite{K76}. Also see \cite{EE87}.

Let $\HH$ be a separable Hilbert space and $h(\cdot,\cdot):\HH\times\HH\ra\C$ be a sesquilinear form. It is called \emph{sectorial} if there exist $\ga\in\R$ and $\tha\in[0,\frac{\pi}{2})$ so that
\begin{equation*}
h(u,u)\in\{z\in\C||\arg(z-\ga)|\leq\tha\}\,\,\text{for any}\,\, u\in\QQ(h)\,\,\text{with}\,\,\|u\|=1,
\end{equation*}
where $\QQ(h)$ is the form domain of $h$. In particular, any symmetric sesquilinear form bounded from below is sectorial. For relatively bounded perturbation, we have (see \cite[Theorem
VI.1.33]{K76})
\begin{thm}\label{relatively-bounded-pert}
Let $h$ be a sectorial form and $h'$ be $h$-bounded, i.e., $\QQ(h)\subset\QQ(h')$ and there exist nonnegative constants $a$ and $b$ such that
\begin{equation*}
|h'(u,u)|\leq ah(u,u)+b\|u\|^{2}\,\,\text{for any}\,\,u\in\QQ(h).
\end{equation*}
If $a<1$, then $h+h'$ is sectorial. $h+h'$ is closable or closed if and only if $h$ is closable or closed, respectively.
\end{thm}

Let $H:\HH\ra\HH$ be a linear operator with domain $\mathscr{D}(H)$. $H$ is said to be \textit{accretive} if $\Re\lan u,Hu\ran\geq0$ for all $u\in\mathscr{D}(H)$. It is
said to be \textit{$m$-accretive} if for any $z\in\mathbb{C}$ with $\Re z>0$, there hold
\begin{equation*}
(H+z)^{-1}\in\mathscr{L}(\HH)\quad\text{and}\quad\|(H+z)^{-1}\|\leq\frac{1}{\Re z},
\end{equation*}
where $\mathscr{L}(\HH)$ denotes the space of all bounded linear operators on $\HH$. It's not hard to see that $m$-accretive operator is maximal accretive in the sense that it is accretive and has no proper accretive extension. If there are $\ga\in\R$ and $\tha\in[0,\frac{\pi}{2})$ so that
\begin{equation*}
\lan
u,Hu\rangle\in\{z\in\C||\arg(z-\ga)|\leq\tha\}\,\,\text{for any}\,\,u\in\mathscr{D}(H)\,\,\text{with}\,\,\|u\|=1,
\end{equation*}
then $H$ is said to be \textit{sectorial}. $H$ is \textit{$m$-sectorial} if it is both $m$-accretive and sectorial.

If $H$ is sectorial, then the sesquilinear form $h(\cdot,\cdot)$ on $\QQ(h)=\mathscr{D}(H)$ defined by
\begin{equation*}
h(u,v)=\lan u,Hv\ran,\quad u,v\in\QQ(h)
\end{equation*}
is sectorial and closable (see \cite[Theorem VI.1.27]{K76}). In particular, any symmetric operator bounded from below defines a closable sectorial form. Conversely, we have (see \cite[Theorem VI.2.1,Theorem V.2.6]{K76})

\begin{thm}\label{representation-theorem}
Let $h(\cdot,\cdot)$ be a densely defined and closed sectorial form in $\HH$ with form domain $\QQ(h)$. Then there exists a unique $m$-sectorial operator $H$ such that $\mathscr{D}(H)\subset\QQ(h)$ and
\begin{equation*}
h(u,v)=\lan
u,Hv\ran\,\,\text{for}\,\,u\in\QQ(h)\,\,\text{and}\,\,v\in\mathscr{D}(H).
\end{equation*}
If, in addition, $h(\cdot,\cdot)$ is symmetric and bounded from below, then the associated $m$-sectorial operator $H$ is self-adjoint with the same lower bounded.
\end{thm}

The second part of the above theorem is well-known and widely used in the theory of Schr\"{o}dinger operator. We also used the following result (see \cite[Lemma VI.3.1]{K76}).

\begin{thm}\label{representation-theorem-1}
Let $h(\cdot,\cdot)$ be a densely defined, symmetric, nonnegative closed form with the associated nonnegative self-adjoint operator $H$. Let $q(\cdot,\cdot)$ be a form relatively bounded with respect
to $h$ so that
\begin{equation*}
|q(u,u)|\leq Ch(u,u),\quad u\in\QQ(h)
\end{equation*}
for some $C\geq0$. Then there is $B\in\mathscr{L}(\HH)$ with $\|B\|\leq\epsilon C$ such that
\begin{equation*}
q(u,v)=\lan\sqrt{H}u,B\sqrt{H}v\ran,\quad u,v\in\QQ(h)=\mathscr{D}(\sqrt{H}),
\end{equation*}
where $\ep=1$ or $2$ according as $q$ is symmetric or not.
\end{thm}


\section{The Helffer-Sj\"{o}strand Formula}\label{app-Helffer-Sjostrand-formula}

In this section, we define the class of slowly decreasing smooth functions and review the Helffer-Sj\"{o}strand formula (see \cite{HS89}), which provides an alternative approach to the spectral
theory of self-adjoint operators. The material below is taken from \cite{Da95}.

\begin{defn}\label{def-slowly-decreasing-smooth-fun}
A function $f$ is said to be in $\mathscr{A}$, the class of slowly decreasing smooth functions on $\R$, if $f\in C^{\infty}(\R)$ and there exit $\mu>0$ and a sequence of constants $c_{n}\geq0$, $n\geq1$ so that
\begin{equation*}
|f^{(n)}(u)|\leq c_{n}\lan u\ran^{-n-\mu},\quad\forall\,\,u\in\R,\,\,\forall\,\,n\geq1,
\end{equation*}
where $\lan u\ran\equiv\sqrt{1+|u|^{2}}$. We define the norms on $\mathscr{A}$: for $f\in\mathscr{A}$,
\begin{equation*}
\interleave f\interleave_{n}=\sum_{r=0}^{n}\int_{\R}|f^{(r)}(u)|\lan u\ran^{r-1}dx,\quad n\geq1.
\end{equation*}
\end{defn}

Let $\tau\in C^{\infty}(\R)$ with $\tau(u)=1$ if $|u|<1$ and $\tau(u)=0$ if $|u|>2$. For $f\in\mathscr{A}$, the smooth (non analytic) extensions $\tilde{f}_{n}:\C\ra\C$
of $f$ are defined by
\begin{equation*}
\tilde{f}_{n}(z)=\bigg\{\frac{1}{r!}\sum_{r=0}^{n}f^{(r)}(u)(iv)\bigg\}\si(u,v),\quad n\geq1,
\end{equation*}
where $z=u+iv$ and $\si(u,v)=\tau\big(\frac{v}{\lan u\ran}\big)$. Define
$\frac{\pa\tilde{f}_{n}(z)}{\pa\bar{z}}=\frac{1}{2}\big\{\frac{\pa\tilde{f}_{n}(z)}{\pa
u}+i\frac{\pa\tilde{f}_{n}(z)}{\pa v}\big\}$. Direct calculation shows
\begin{equation*}
\frac{\pa\tilde{f}_{n}(z)}{\pa\bar{z}}=\frac{1}{2}\bigg\{\sum_{r=0}^{n}\frac{1}{r!}f^{(r)}(u)(iv)^{r}\bigg\}\big(\sigma_{u}(u,v)+\sigma_{v}(u,v)\big)+\frac{1}{2n!}f^{(n+1)}(u)(iv)^{n}\sigma(u,v).
\end{equation*}
Obviously, $\si(u,v)=0$ if $|v|\geq2\lan u\ran$ and both $\si_{u}(u,v)=0$ and $\si_{v}(u,v)=0$ if $|v|\leq\lan u\ran$ or $|v|\geq2\lan u\ran$. Thus, by introducing the sets $U=\{(u,v)\in\R^{2}|\lan u\ran<|v|<2\lan u\ran\}$ and $V=\{(u,v)\in\R^{2}|0<|v|<2\lan u\ran\}$, we have
\begin{equation}\label{app-an-inequality}
\bigg|\frac{\pa\tilde{f}_{n}(z)}{\pa\bar{z}}\bigg|\leq
C\bigg\{\sum_{r=0}^{n}\frac{1}{r!}|f^{(r)}(u)|\frac{|v|^{r}}{\lan
u\ran}\bigg\}\chi_{U}(u,v)+\frac{1}{2n!}|f^{(n+1)}(u)||v|^{n}\chi_{V}(u,v)
\end{equation}
for some $C>0$ only depending on $\tau$, where $\chi_{U}$ and $\chi_{V}$ are characteristic functions for $U$ and $V$, respectively.

\begin{thm}[\cite{Da95}]\label{theorem-helffer-sjostrand}
Let $f\in\mathscr{A}$ and $H$ be a self-adjoint on a separable Hilbert space. Then the integral
\begin{equation*}
\int_{\R^{2}}\frac{\pa\tilde{f}_{n}(z)}{\pa\bar{z}}(H-z)^{-1}dudv
\end{equation*}
converges in operator norm and is independent of $n$ and $\tau$. Moreover,
\begin{equation*}
\bigg\|\int_{\R^{2}}\frac{\pa\tilde{f}_{n}(z)}{\pa\bar{z}}(H-z)^{-1}dudv\bigg\|\leq
c\interleave f\interleave_{n+1},\quad\forall\,\,n\geq1,
\end{equation*}
where $c>0$ is a constant independent of $f$ and $n$.
\end{thm}

It should be pointed out the fact that the constant $c$ is independent of $n$ is due to Germinet and Klein \cite{GK03}. This follows from the fact that $\frac{2^{n}}{n!}\ra0$ as $n\ra\infty$.

We then define for $f\in\mathscr{A}$
\begin{equation}\label{the-Helffer-Sjostrand-formula}
f(H)=\frac{1}{\pi}\int_{\R^{2}}\frac{\pa\tilde{f}_{n}(z)}{\pa\bar{z}}(H-z)^{-1}dudv,
\end{equation}
which is referred to as the Helffer-Sj\"{o}strand formula. By Theorem \ref{theorem-helffer-sjostrand}, $\|f(H)\|\leq c\interleave f\interleave_{n+1}$ for all $n\geq1$.

\end{document}